\documentclass[10pt]{amsart}

\usepackage{amsfonts}

\newtheorem{Theorem}{Theorem}

\newtheorem{Definition}{Definition}
\newtheorem{Proposition}{Proposition}

\newtheorem{Remark}{Remark}
\newtheorem{Example}{Example}

\newcommand{\LL}{{\mathrm{L}}}
\newcommand{\Id}{{\mathbf{1}}}
\newcommand{\dom}{{\mathrm{dom}~}}
\newcommand{\img}{{\mathrm{rng}~}}
\newcommand{\N}{\ensuremath{{\mathbb N}}}

\newcommand{\R}{\ensuremath{{\mathbb R}}}

\newcommand{\Imm}{{\mathrm{Im}~}}
\newcommand{\Ree}{{\mathrm{Re}~}}

\newcommand{\hil}{\mathcal H}
\newcommand{\ra}{\rangle}
\newcommand{\la}{\langle}
\newcommand{\wgconv}{\stackrel{\mathrm{W\Gamma}}{\longrightarrow}}
\newcommand{\sgconv}{\stackrel{\mathrm{S\Gamma}}{\longrightarrow}}
\newcommand{\wseta}{\rightharpoonup}

\begin{document}

\title[Quantum resolvent and $\Gamma$ con\-ver\-gences]{ Quantum singular operator limits of thin Dirichlet tubes via $\Gamma$-convergence}
\author{C\'{e}sar R. de Oliveira}
\thanks{The author thanks Profs.\ R.\ Froese and J.\ Feldman for discussions, the warm hospitality of
PIMS  and UBC (Vancouver), and the partial support by CAPES (Brazil).}
\address{Departamento de Matem\'{a}tica -- UFSCar, S\~{a}o Carlos, SP, 13560-970 Brazil}
\email{oliveira@ufscar.br}


\begin{abstract} The $\Gamma$-con\-ver\-gence of lower bounded quadratic
forms is  used to study the singular operator limit of thin tubes (i.e., the vanishing
of the cross section diameter) of the Laplace operator with Dirichlet boundary conditions;  a
procedure to obtain the effective Schr\"odinger operator (in different subspaces) is proposed,
generalizing recent results in case of compact tubes. Finally, after scaling curvature and torsion the limit of a broken line is briefly investigated.  
\end{abstract}

\maketitle

\tableofcontents

\section{Introduction} \label{secIntro} Among the tools for studying limits of self-adjoint operators
in Hilbert spaces are the resolvent con\-ver\-gence and the sesqui\-lin\-e\-ar form con\-ver\-gence.
In the context of quantum mechanics, it is well known that in case of monotone sequences of operators 
these approaches are strictly related, as discussed, for example, in Section~10.4 of~\cite{ISTQD} and
Section~VIII.7 (Supplementary Material) of~\cite{RS1}.  These form con\-ver\-gences have the advantage
of dealing with some singular limits in quantum mechanics, which is well exemplified by mathematical
arguments supporting the Aharonov-Bohm Hamiltonian \cite{MVG,deOPer}.   

Let $\Id$ denote the identity
operator and $0 \le T_j$ be a sequence of positive (or uniformly lower bounded in general)
self-adjoint operators acting in the Hilbert space $\hil$. From the technical point of view
what happens is that the monotone increasing of the sequence of resolvent operators
$R_{-\lambda}(T_j):= (T_j+\lambda\Id)^{-1}$ ($\lambda>0$) implies the monotone decreasing of the
corresponding sequence of sesqui\-lin\-e\-ar forms and vice versa. Due to monotonicity, in such cases
one clearly understands the existence of limits and have some insight in limit form domains. However,
in principle it is not at all clear what are the relations between strong resolvent con\-ver\-gence of
operators and sesqui\-lin\-e\-ar form con\-ver\-gence in more general cases, say, if one requires that
the self-adjoint operators or closed forms are only uniformly bounded from below. It happens that such
relations are well known among people interested in variational con\-ver\-gences and applied
mathematics, and it is directly related to the concept of $\Gamma$-con\-ver\-gence; but when
addressing sesqui\-lin\-e\-ar forms these variational problems are usually formulated in real Hilbert
spaces and the theory has been developed under this condition. However, in quantum mechanics the Hilbert spaces are usually complex, and an adaptation of the main results to complex Hilbert spaces will appear elsewhere \cite{deOcomplex}.    Due to the particular class of applications we have in mind, here it will be enough to reduce some key arguments to real Hilbert spaces.

In Section~\ref{sectionGammaConv} the very basics of $\Gamma$-convergence are recalled in a suitable way; for details the reader is referred to the important 
monographs in the field  \cite{DalMaso,Braides}.

As an application of $\Gamma$-con\-ver\-gence of forms we study the
limit operator obtained from the Laplacian (with Dirichlet boundary conditions) restricted to a tube
in $\R^3$ that shrinks to a smooth curve. This problem has been considered in the interesting work
\cite{BMT}, where the strong $\Gamma$-con\-ver\-gence was employed to study the limit operator and
con\-ver\-gence of eigenvalues and eigenvectors, but the curve was supposed to have finite length $L$
and the con\-ver\-gence restricted to the subspace of vectors of the form $w(s)u_0(y)$ (``the first
sector''), with
$w\in \hil^1_0[0,L]$ (i.e., the usual Sobolev space)  and $u_0$ being the first eigenfunction of the
restriction of the Laplacian to the tube cross section ($s$ denotes the curve arc length and
$y=(y_1,y_2)$ the cross section variables). Here we indicate how the proofs in \cite{BMT} can be worked out to get an effective operator in case of
curves of infinite length; in fact it will be necessary to go a step further and we prove both strong
and weak
$\Gamma$-con\-ver\-gences of forms which imply the strong resolvent con\-ver\-gence of operators to an
effective Schr\"odinger operator on the curve. We also give an alternative proof of the spectral
con\-ver\-gence discussed in
\cite{BMT} in case of finite length curves:  compactness arguments due to the  boundedness of the tube
in this case will be essential; our proof also clarifies the mechanism behind the spectral
con\-ver\-gence. These issues are discussed in Section~\ref{sectionLimitCurve}.

The use of quadratic forms, and the possibility of considering $\Gamma$-con\-ver\-gence, in the
application mentioned above is important because due to the intricate geometry of the allowed tubes
the expressions of the actions of the involved operators are rather complicated, with the presence of
mixed derivatives and nonlinear coefficients. In case the curve lives in
$\R^2$ there are results about the limit operator in
\cite{DellAntTen,ACF} and the effective potential is written in terms of the curvature; the main
novelty in case of
$\R^3$ considered in \cite{BMT} is the additional presence of ``twisting'' and torsion in the
effective potential, since the case of untwisted tubes has also been previously studied in~\cite{DE}
(see also~\cite{ClaBra,FreitasKrej}).  Since in both
$\R^2$ and untwisted tubes in $\R^3$ we have simpler expressions for the involved operators, it is
possible to deal directly with the strong resolvent con\-ver\-gence and consider more general spaces
than that related only to the first eigenfunction of the laplace operator in the cross section (i.e.,
the first sector). So, based on the
$\Gamma$-con\-ver\-gence of sesqui\-lin\-e\-ar forms, we propose here a rather natural procedure to
handle the strong con\-ver\-gence in case of vectors of the form $w(s)u_n(y)$, i.e., the $(n+1)$th
sector spanned by the general eigenvector $u_n$ of the Laplacian restricted to the cross section,
and we will impose that each eigenvalue of this operator is simple in order to simplify the implementation
of the
$\Gamma$-con\-ver\-gence of  $b_n^\varepsilon$ and the strong resolvent
con\-ver\-gence of the associated self-adjoint operators $H^\varepsilon_n$, $n\ge1$, which are
introduced in Definition~\ref{defHneps}. It turned out that the resulting effective operator depends
on the sector considered, an effect not present in
$\R^2$
\cite{DellAntTen,ACF}. This ``new'' effect is understood since in $\R^2$ there is a separation of
variables
$s$ and $y$ (see the top of page~14 of \cite{ACF}, but be aware of the different notations), which in
general does not occur in $\R^3$ and so the necessity of introducing a specific procedure. These
developments are also presented in  Section~\ref{sectionLimitCurve}. Roughly speaking, from the
dynamical point of view the dependence of the limit operator on the sector of the Hilbert space would
correspond to the dependence of the effective potential on the initial condition, a phenomenon already
noticed in~\cite{FH,Takens}.   

It must be mentioned that  the author has no intention to claim that
$\Gamma$-con\-ver\-gence should replace the traditional operator techniques in such kind of problems
with no separation of variables \cite{BF,FS1,FS2,Krej2}, including some important spectral results for
twisted tubes~\cite{EKK}. Further, we need a combination of strong and weak $\Gamma$-convergences of
quadratic forms in order to get strong resolvent con\-ver\-gence of the related operators, and other
assumptions and tools must be added to get norm resolvent con\-ver\-gence (see, for instance,
Subsection~\ref{subsecBTubes}). The $\Gamma$-con\-ver\-gence in this context must be seen as just
another available tool for  mathematical physicists.   

With the effective limit operator at hand, we finally
discuss the limit of a smooth curve approaching a broken line, similarly to \cite{DellAntTen,ACF}, but
here we confine ourselves to take one limit at a time, that is, first we constrain the particle motion
from the tube to the curve (as discussed above), then we take the limit when the curve approaches the
broken line; we shall closely  follow a discussion in
\cite{ACF}. Besides the negative term related to the curvature, in our case the effective potential
$V^{\mathrm{eff}}(s)$ has the additional presence of a positive term related to ``twist'' and torsion
so that the technical condition 
\[
\int_\R V^{\mathrm{eff}}(s)ds\ne0
\] assumed in \cite{ACF} might not hold in some cases. We have then checked that it is still possible
to follow the same proofs, but with suitable adaptations, even if the above integral vanishes. The
bottom line is that different self-adjoint realizations for the operator on the broken line are found.
These broken-line limits are shortly discussed in Section~\ref{sectBrokenLine}.

Of course in this introduction we have skipped many technical details, some of them are fundamental to
a correct understanding of the contents of this work. In the following sections I shall try to fill
out those gaps.  I expect  this work will motivate researchers to seriously consider the
$\Gamma$-con\-ver\-gence of forms as a useful way to study (singular) limits of observables in quantum
mechanics.

\section{Strong Resolvent and $\Gamma$ Convergences}\label{sectionGammaConv}
\subsection{$\Gamma$-Convergence}  Our sequences (more properly they should be called {\rm families})
of self-adjoint operators $T_\varepsilon$, with domain $\dom T_\varepsilon$ in a  separable
Hilbert space $\hil$, and the corresponding closed sesqui\-lin\-e\-ar forms
$b_\varepsilon$ will be indexed by the parameter $\varepsilon>0$ and, by definiteness, we think of the
limit
$\varepsilon\to0$ and we want to study the limit $T$ (resp.\  $b$) of
$T_\varepsilon$ (resp.\  $b_\varepsilon$).  The domain of $T$ will not be supposed to be dense
in~$\hil$ and its closure will be denoted by $\hil_0=\overline{\dom T}$ (with $\img T\subset
\hil_0$); usually this is indicated by simply saying that ``$T$ is self-adjoint in~$\hil_0$.''

We assume that  a sesqui\-lin\-e\-ar form $b(\zeta,\eta)$ is linear in the second entry and antilinear
in the first one.   As usual the real-valued function $\zeta\mapsto b(\zeta,\zeta)$ will be simply
denoted by
$b(\zeta)$ and called the associated quadratic form; we will use the terms sesqui\-lin\-e\-ar and
quadratic forms almost interchangeably, since usually the context makes it clear  which one is being
referred to.  It will also be assumed that $b$ is positive (or lower bounded in general) and
$b(\zeta)=\infty$ if $\zeta$ does not belong to its domain
$\dom  b$; this is important in order to guarantee that in some cases $b$ is lower semicontinuous,
which is equivalent to $b$ be the sesqui\-lin\-e\-ar form generated by a positive self-adjoint
operator $T$, that is, 
\[ b(\zeta,\eta)=\la T^{1/2}\zeta,T^{1/2}\eta\ra, \quad \zeta,\eta\in\dom b=\dom T^{1/2};
\]see Theorem 9.3.11 in \cite{ISTQD}. By allowing $b(\zeta)=\infty$ we have a handy way to work in the
larger space
$\hil$ instead of only in~$\hil_0=\overline{\dom T}$. If
$\lambda\in\R$, then $b+\lambda$ indicates the sesqui\-lin\-e\-ar form
$(b+\lambda)(\zeta,\eta):=b(\zeta,\eta)+\lambda \la\zeta,\eta\ra$, whose corresponding quadratic form
is
$b(\zeta)+\lambda\|\zeta\|^2$.

It is known that (Lemma~10.4.4 in \cite{ISTQD}), for any $\lambda>0$, one has
$b_{\varepsilon_1}\le b_{\varepsilon_2}$ iff $R_{-\lambda}(T_{\varepsilon_2})\le
R_{-\lambda}(T_{\varepsilon_1})$, so that a sequence of quadratic forms is monotone iff the
corresponding sequence of resolvent operators is monotone. This has been explored in the quantum
mechanics literature in order to get strong resolvent limits of self-adjoint operators through the
study of quadratic forms (see, for instance, Section~10.4 of~\cite{ISTQD} and Section~VIII.7
of~\cite{RS1}). For more general sequences of operators, the form counterpart of the strong resolvent
con\-ver\-gence is not so direct and it was found that the correct concept comes from the so-called
$\Gamma$-con\-ver\-gence \cite{DalMaso}. In what follows the concept of
$\Gamma$-con\-ver\-gence will be recalled in a suitable way, and then applied to the study of singular limits of Dirichlet tubes in other sections.

The general concept of $\Gamma$-con\-ver\-gence is not restricted to quadratic forms and can be
applied to quite general topological spaces, but in this section the ideas will mostly be suitably
adapted to our framework; e.g., we try to restrict the discussion to Hilbert spaces and to lower
semicontinuous functions, since the quadratic forms we are interested in have this property. The
general theory is nicely presented in the book
\cite{DalMaso}, to which we will often  refer. $\hil$ always denote a separable Hilbert space and
$B(\zeta;\delta)$ the open ball cantered at $\zeta\in\hil$ of radius $\delta>0$; finally $\overline\R
:=\R\cup\{\infty\}$ and the symbol l.sc.\ will be a shorthand to {\em lower semicontinuous}.

\begin{Definition}\label{defGamma} The lower $\Gamma$-limit  of a sequence of l.sc.\ functions
$f_\varepsilon:\hil\to\overline \R$ is the function $f^-:\hil\to \overline\R$ given by
\[ f^-(\zeta) = \lim_{\delta\to0}\; \liminf_{\varepsilon\to0}\; \inf \left\{ f_\varepsilon (\eta):
\eta\in B(\zeta;\delta) \right\},\quad \zeta\in\hil.
\]The upper  $\Gamma$-limit $f^+(\zeta)$ of $f_\varepsilon$ is defined by replacing $\liminf$ by
$\limsup$ in the above expression. If $f^-=f^+=:f$ we say that such function is the
$\Gamma$-limit of $f_\varepsilon$
 and it will be denoted by
\[ f=\Gamma\text{-}\lim_{\varepsilon\to0}f_\varepsilon.
\]
\end{Definition}

\begin{Remark}\label{remark1} It was assumed in Definition~\ref{defGamma} that the topology of $\hil$ is the usual
norm topology, and in this case we speak of {\em strong $\Gamma$-con\-ver\-gence}. If the weak
topology is considered, the balls
$B(\zeta;\delta)$ must be replaced by the set of all open weak neighborhoods of
$\zeta$ \cite{DalMaso}, and in this case we speak of {\em weak
$\Gamma$-con\-ver\-gence}. Both concepts will be important here, and in general they are not
equivalent since the norm is not continuous in the weak topology (see Example~6.6 in
\cite{DalMaso}). When convenient, the symbols \[ f_\varepsilon\sgconv f,\quad f_\varepsilon\wgconv f
\] will be used to indicate that $f_\varepsilon$ $\Gamma$-converges to $f$ in the strong and weak
sense in
$\hil$, respectively. See also Proposition~\ref{GammaAltern} and Remark~\ref{remark3}.
\end{Remark}

\begin{Example} The sequence $f_\varepsilon:\R\to\R$, $f_\varepsilon(x)=\sin(x/\varepsilon)$,
$\Gamma$-converges to the constant function $f(x)=-1$ as $\varepsilon\to0$. This simple example nicely
illustrates the property of ``con\-ver\-gence of minima'' that motivated the introduction of the
$\Gamma$-con\-ver\-gence. In quantum mechanics one is used to the con\-ver\-gence of averages, so that
the natural guess (if any) for the limit in this example would be the null function. 
\end{Example}

\begin{Remark} The $\Gamma$-con\-ver\-gence is usually different from both the pointwise and weak
limit of functions and roughly it can be illustrated as follows. Assume one is studying a
heterogeneous material subject to strong tensions whose intensity in some regions is measured by the
parameter $1/\varepsilon$ and it will undergo a kind of phase transition as
$\varepsilon\to0$; for each $\varepsilon>0$ one computes the equilibrium configuration of this
material via a minimum of certain energy functional;  this transition would be computed as the limit
of the equilibria (i.e., minima of such functionals), and this limit (i.e., minimum of an ``effective
limit functional'') would be quite singular. In many instances
$\Gamma$-con\-ver\-gence is the correct concept to describe such situations \cite{Braides,DalMaso} and
its importance in the asymptotics of variational problems relies here. 
\end{Remark}

\begin{Remark}From some points of view the notion of $\Gamma$-con\-ver\-gence is quite subtle, as
exemplified by the following facts:
\begin{itemize}
\item[a)] Due to the prominent role played by minima, in general 
\[\Gamma\text{-}\lim_{\varepsilon\to0}f_\varepsilon\ne
-\left(\Gamma\text{-}\lim_{\varepsilon\to0}(-f_\varepsilon)\right).
\]
\item[b)] Assume that $f=\Gamma\text{-}\lim_{\varepsilon\to0}f_\varepsilon$ and
$g=\Gamma\text{-}\lim_{\varepsilon\to0}g_\varepsilon$; it may happen that
$(f_\varepsilon+g_\varepsilon)$ is not $\Gamma$-convergent.
\item[c)] Clearly it is not necessary to restrict the definition of $\Gamma$-con\-ver\-gence to l.sc.\
functions. If
$f_\varepsilon=f$, for all $\varepsilon$, and $f$ is not l.sc., then
$\Gamma\text{-}\lim_{\varepsilon\to0}f$ is the greatest lower semicontinuous function majorized by~$f$ (the so-called l.sc.\ envelope of $f$), and so different from~$f$! 
\end{itemize}
\end{Remark}

Since $\hil$ satisfies the first axiom of countability, Proposition~8.1 of \cite{DalMaso} implies the
following handy characterization of strong and weak $\Gamma$-con\-ver\-gence: 

\begin{Proposition}\label{GammaAltern} The sequence $f_\varepsilon:\hil\to \overline{\R}$ strongly
$\Gamma$-converges to $f$ (that is, $f_\varepsilon\sgconv f$) iff the following two conditions are
satisfied:
\begin{itemize}
\item[i)] For every $\zeta\in\hil$ and every $\zeta_\varepsilon\to\zeta$ in $\hil$ one has
\[ f(\zeta)\le \liminf_{\varepsilon\to0} f_\varepsilon(\zeta_\varepsilon).
\]
\item[ii)] For every $\zeta\in\hil$ there exists a sequence $\zeta_\varepsilon\to\zeta$ in
$\hil$ such that
\[ f(\zeta)=\lim_{\varepsilon\to0} f_\varepsilon(\zeta_\varepsilon).
\]
\end{itemize}
\end{Proposition}

\begin{Remark}\label{remark3} If instead of strong con\-ver\-gence $\zeta_\varepsilon\to\zeta$ one
considers weak con\-ver\-gence $\zeta_\varepsilon\wseta\zeta$ in Proposition~\ref{GammaAltern}, then
we have a characterization of
$f_\varepsilon\wgconv f$. Here this characterization will be used in practice,  and so it justifies
the lack of details with respect to weak $\Gamma$-con\-ver\-gence in Remark~\ref{remark1}.
\end{Remark}

Recall that a function $f:\hil\to\overline\R$ is {\em coercive} if for every $x\in\R$ the set
$f^{-1}(-\infty,x]$ is precompact in~$\hil$. Since $\hil$ is reflexive, it turns out that a function
$f$ is coercive in the weak topology of $\hil$ iff
$\lim_{\|\zeta\|\to\infty}f(\zeta)=\infty$.
 A sequence of functions $f_\varepsilon:\hil\to\overline\R$ is {\em equicoercive} if there exists a
coercive
$\varphi:\hil\to\overline\R$ such that $f_\varepsilon\ge\varphi$, for all~$\varepsilon>0$. The
following results, stated as Theorems~\ref{GammaMinimizers} and~\ref{GammaPerturb2}, will be useful
later on. The first one  is about con\-ver\-gence of minimizers  (for the proof see~\cite{DalMaso},
Chapter~7), whereas the second one gives some conditions guaranteeing that the
$\Gamma$-con\-ver\-gence is stable under continuous perturbations (see Proposition~6.21 in the
book~\cite{DalMaso}.).

\begin{Theorem}\label{GammaMinimizers}
 Assume that $f_\varepsilon:\hil\to\overline\R$ $\Gamma$-converges to $f$ and let
$\zeta_\varepsilon$ be a minimizer of $f_\varepsilon$, for all $\varepsilon$. Then any cluster point of
$(\zeta_\varepsilon)$ is a minimizer of~$f$.  Further,  if $f_\varepsilon$ is equicoercive and $f$ has
a unique minimizer $x_0$, then $\zeta_\varepsilon$ converges to~$x_0$.
\end{Theorem}

\begin{Theorem}\label{GammaPerturb2} Let $b_\varepsilon,b\ge\beta>-\infty$ be closed and (uniformly)
lower bounded sesqui\-lin\-e\-ar forms in~$\hil$ and $f:\hil\to\R$ a continuous function. Then
$b_\varepsilon$
$\Gamma$-converges to $b$ iff $(b_\varepsilon+f)$
$\Gamma$-converges to $(b+f)$. In particular, in case of both weak and strong
$\Gamma$-con\-ver\-gences in $\hil$, it holds for the functional $f(\cdot)=\la\eta,\cdot\ra +
\la\cdot,\eta\ra$, defined for each fixed $ \eta\in\hil$.
\end{Theorem}

\subsection{$\Gamma$ and Resolvent Convergences} Now we recall the main results with respect to the
relation between the strong resolvent con\-ver\-gence of self-adjoint operators and
$\Gamma$-con\-ver\-gence of the associated  sesqui\-lin\-e\-ar forms \cite{DalMaso,deOcomplex}.  
\begin{Theorem}\label{mainTheorGamma} Let $b_\varepsilon,b$ be positive (or uniformly lower bounded)
closed sesqui\-lin\-e\-ar forms in the Hilbert space $\hil$, and
$T_\varepsilon,T$ the corresponding associated positive self-adjoint operators. Then the following
statements are equivalent:
\begin{itemize}
\item[i)] $b_\varepsilon\sgconv b$ and, for each $\zeta\in\hil$, $b(\zeta)\le
\liminf_{\varepsilon\to0} b_\varepsilon(\zeta_\varepsilon)$, $\forall \zeta_\varepsilon\wseta
\zeta$ in $\hil$.
\item[ii)]$b_\varepsilon\sgconv b$ and $b_\varepsilon\wgconv b$.
\item[iii)] $b_\varepsilon+\lambda \sgconv b+\lambda$ and $b_\varepsilon+\lambda \wgconv b+\lambda$,
for some
$\lambda>0$ (and so for all $\lambda\ge0$).
\item[iv)] For all $\eta\in\hil$ and $\lambda>0$, the sequence
\[
\min_{\zeta\in\hil}\left[b_\varepsilon(\zeta)+\lambda\|\zeta\|^2+\frac12 (\la\eta,\zeta\ra
+\la\zeta,\eta\ra)\right]
\] converges to
\[
\min_{\zeta\in\hil}\left[b(\zeta)+\lambda\|\zeta\|^2+\frac12 (\la\eta,\zeta\ra
+\la\zeta,\eta\ra)\right].
\]
\item[v)] $T_\varepsilon$ converges to $T$ in the strong resolvent sense in
$\hil_0=\overline{\dom T}\subset\hil$, that is, 
\[
\lim_{\varepsilon\to0} R_{-\lambda}(T_\varepsilon)\zeta = R_{-\lambda}(T)P_0\zeta,\quad
\forall\zeta\in\hil,\forall \lambda>0,
\]where $P_0$ is the orthogonal projection onto~$\hil_0$.
\end{itemize}
\end{Theorem}

 The next two results are included mainly because they help to elucidate the connection
between forms, operator actions  and domains on the one hand, and minimalization of suitable
functionals on the other hand; this sheds some light on  the role played by
$\Gamma$-con\-ver\-gence in the con\-ver\-gence of self-adjoint operators. 

\begin{Proposition}\label{propMudancaComplex} Let $b\ge0$ be a closed sesqui\-lin\-e\-ar form on the
 Hilbert space~$\hil$, $T\ge0$ the self-adjoint operator associated with $b$ and
$P_0$ be the orthogonal projection onto $\hil_0=\overline{\dom T}\subset\hil$. Then
$\zeta\in\dom T$ and $T\zeta=P_0\eta$ iff $\zeta$ is a minimum point (also called minimizer) of the
functional
\[ g:\hil\to\overline\R,\qquad g(\zeta)= b(\zeta) - \la\eta,\zeta\ra-\la\zeta,\eta\ra.
\]
\end{Proposition}

\begin{Proposition}\label{propCharacQuadra} Let  $T:\dom T\to\hil$ be a positive self-adjoint operator,
$\overline{\dom T}= \hil_0$, and $b^T:\hil\to\overline\R$ the quadratic form generated by~$T$. Then
\begin{eqnarray*} b^T(\zeta)&=&\sup_{\eta\in\dom T} \left[ \la T\eta,\zeta\ra +\la\zeta,T\eta\ra - \la
T\eta,\eta\ra
\right] 
\\ &=& \sup_{\eta\in\dom T} \left[ b^T(\eta)+ \la T\eta,\zeta\ra +\la\zeta,T\eta\ra - 2\la
T\eta,\eta\ra
\right],
\end{eqnarray*}for all $\zeta\in \hil_0$ and $b^T(\zeta)=\infty$ if $\zeta\in
\hil\setminus\hil_0$.
\end{Proposition}

Finally we recall Theorem~13.5 in
\cite{DalMaso}: 

\begin{Theorem}\label{GammaWeakRW} Let $b_\varepsilon,b\ge \beta>0$ be sesqui\-lin\-e\-ar forms on the Hilbert space $\hil$ and $T_\varepsilon,T\ge\beta\Id$ the corresponding associated self-adjoint operators,
and let
$\overline{\dom T}=\hil_0\subset\hil$. Then the  following  statements are equivalent: 
\begin{itemize}
\item[i)] $b_\varepsilon \wgconv b$.
\item[ii)] $R_0(T_\varepsilon)$ converges weakly to $R_0(T)P_0$, where $P_0$ is the orthogonal
projection onto~$\hil_0$.
\end{itemize} 
\end{Theorem}

\subsection{Norm Resolvent Convergence} Before turning to an application of
Theorem~\ref{mainTheorGamma} in the next section, we introduce additional conditions in order to get
norm resolvent con\-ver\-gence of operators from
$\Gamma$-con\-ver\-gence. This condition will be used to recover a spectral con\-ver\-gence proved
in~\cite{BMT} and, from the technical point of view, can be considered our first contribution.
  
\begin{Proposition}\label{GammaNorm} Let $b_\varepsilon,b\ge \beta>-\infty$ be closed
sesqui\-lin\-e\-ar forms and
$T_\varepsilon,T\ge\beta\Id$ the corresponding associated self-adjoint operators, and let
$\overline{\dom T}=\hil_0\subset\hil$. Assume that the following three conditions hold:
\begin{itemize}
\item[a)] $b_\varepsilon\sgconv b$ and $b_\varepsilon\wgconv b$. 
\item[b)] The resolvent operator $R_{-\lambda}(T)$ is compact in $\hil_0$ for some real number
$\lambda>|\beta|$.
\item[c)] There exists a Hilbert space $\mathcal K$,  compactly embedded in~$\hil$, so that if the sequence
$(\psi_\varepsilon)$ is bounded in $\hil$ and  $(b_\varepsilon(\psi_\varepsilon))$ is also
bounded, then
$(\psi_\varepsilon)$ is a bounded subset of~$\mathcal K$.
\end{itemize} Then, $T_\varepsilon$ converges in norm resolvent sense to $T$ in $\hil_0$ as
$\varepsilon\to0$.
\end{Proposition}
\begin{proof} We must show that $R_{-\lambda}(T_\varepsilon)$ converges in operator norm to
$R_{-\lambda}(T)P_0$, where $P_0$ is the orthogonal projection onto~$\hil_0$; to simplify the notation
the projection
$P_0$ will be ignored.   

If $R_{-\lambda}(T_\varepsilon)$ does not converge in norm to
$R_{-\lambda}(T)$, there exist $\delta_0>0$ and vectors
$\eta_\varepsilon$, $\|\eta_\varepsilon\|=1$, for a subsequence (we tacitly keep the same notation
after taking subsequences) of indices $\varepsilon\to0$ so that
\[
\left\| R_{-\lambda}(T_\varepsilon)\eta_\varepsilon - R_{-\lambda}(T)\eta_\varepsilon 
\right\|\ge\delta_0, \quad \forall \varepsilon>0.
\] We will argue to  get a contradiction with this inequality, so proving the proposition. Denote
$\zeta_\varepsilon:=R_{-\lambda}(T_\varepsilon)\eta_\varepsilon$.  By the reflexivity of $\hil$  one
can suppose that
$\eta_\varepsilon\wseta \eta$, for some $\eta\in\hil$, and since $R_{-\lambda}(T)$ is compact we have
$R_{-\lambda}(T)\eta_\varepsilon\to R_{-\lambda}(T)\eta$ in $\hil$.   The general inequalities
\[
\|R_{-\lambda}(T_\varepsilon)\eta_\varepsilon\|\le \frac1{|\beta-\lambda|} \mathrm{\quad and\quad}
\|T_\varepsilon R_{-\lambda}(T_\varepsilon)\eta_\varepsilon\|\le\|\eta_\varepsilon\|=1,\quad \forall
\varepsilon,
\] imply that 
\[ |b_\varepsilon(\zeta_\varepsilon)| = |\la T_\varepsilon\zeta_\varepsilon,\zeta_\varepsilon\ra| \le 
\| T_\varepsilon\zeta_\varepsilon\|\; \|\zeta_\varepsilon\|\;\le  \frac1{|\beta-\lambda|},\quad
\forall \varepsilon,
\] and so it follows by \rm{c)} that 
$(R_{-\lambda}(T_\varepsilon)\eta_\varepsilon)$ is a  bounded sequence in $\mathcal K$, and since this
space is compactly embedded in $\hil$ there exists a (strongly) convergent subsequence so that 
\[ R_{-\lambda}(T_\varepsilon)\eta_\varepsilon\to \zeta
\] for some $\zeta\in\hil$. Next we will  employ the $\Gamma$-con\-ver\-gence to show that
$\zeta=R_{-\lambda}(T)\eta$.   

By Proposition~\ref{propMudancaComplex}, for each $\varepsilon$ fixed,
$\zeta_\varepsilon$ is the minimizer in $\hil$ of the functional
\[ g_\varepsilon(\phi) = b_\varepsilon(\phi) + \lambda \|\phi\|^2 -
\la\eta_\varepsilon,\phi\ra -\la\phi,\eta_\varepsilon\ra,
\]whereas $\zeta=R_{-\lambda}(T)\eta$ is the unique minimizer of
\[ g(\phi) = b(\phi) + \lambda \|\phi\|^2 - \la\eta,\phi\ra -\la\phi,\eta\ra.
\]

Since $\lambda>|\beta|$ it follows that $g_\varepsilon$ is weakly equicoercive (see the discussion
after Remark~\ref{remark3}) and, by Theorem~\ref{GammaMinimizers},
$\zeta_\varepsilon\wseta \zeta$. By this weak con\-ver\-gence and
Theorem~\ref{mainTheorGamma}~\rm{iv)}, for all
$\mu\ge \lambda$, one has
\[ b(\zeta) + \mu \|\zeta\|^2 - \la\eta,\zeta\ra -\la\zeta,\eta\ra = \lim_{\varepsilon\to0}
[b_\varepsilon(\zeta_\varepsilon) + \mu \|\zeta_\varepsilon\|^2 - \la\eta,\zeta_\varepsilon\ra
-\la\zeta_\varepsilon,\eta\ra].
\]  Theorem~\ref{mainTheorGamma} and again  $\zeta_\varepsilon\wseta\zeta$ imply 
\[ b(\zeta)+\lambda \|\zeta\|^2 \le \liminf_{\varepsilon\to0} [b_\varepsilon(\zeta_\varepsilon) +
\lambda\|\zeta_\varepsilon\|^2],
\]whereas the lower semicontinuity of the norm with respect to weak con\-ver\-gence gives 
\[ (\mu-\lambda)\|\zeta\|^2 \le  \liminf_{\varepsilon\to0} (\mu-\lambda)\|\zeta_\varepsilon\|^2,\quad
\mu>\lambda.
\]The last three relations imply $\|\zeta\|=\lim_{\varepsilon\to0}\|\zeta_\varepsilon\|$, and together
with
$\zeta_\varepsilon\wseta\zeta$ one obtains the strong con\-ver\-gence
$\zeta_\varepsilon\to\zeta$, that is, 
\[ R_{-\lambda}(T_\varepsilon)\eta_\varepsilon\to R_{-\lambda}(T)\eta,
\] which contradicts the existence of $\delta_0>0$ above. The proof of the proposition is complete.
\end{proof}

\section{Singular Limit of Dirichlet Tubes}\label{sectionLimitCurve} There are many occasions in
which  particles or waves are restricted  to propagate in thin domains along one-dimensional
structures, as a graph. Optical fibers, carbon nanotubes and the motion of valence electrons in
aromatic molecules are good examples. One natural theoretical consideration is to neglect the small
transversal sections and  model these  systems by the true one-dimensional versions by means of
effective parameters and potentials. Besides the necessity of finding effective models, one has also
to somehow decouple the  transversal and the longitudinal variables. Such  confinements can be
realized by strong potentials (see
\cite{FH,WT} and references therein) or  boundary conditions, and the graph can be imbedded in spaces
of different dimensions.     

In this section we apply the $\Gamma$-con\-ver\-gence in  Hilbert spaces to study the limit 
operator of the Dirichlet  Laplacian in a sequence of tubes in $\R^3$ that is squeezed to a curve. 
This kind of problem (mainly in
$\R^2$) has been considered in some papers (for instance,
\cite{ACF,BMT,DellAntTen,DE, FreitasKrej}), and here we show how part of the construction in
\cite{BMT} can be carried out to  curves of infinite length. As already mentioned, since our main interest is in quantum mechanics, in principle one should use complex Hilbert spaces and adapt the results of $\Gamma$-convergence to this more general setting \cite{deOcomplex} before presenting applications, but since the Laplacian is a real operator, and its eigenfunctions can always be supposed to be real valued, one can reduce the arguments to the real setting.

In a second step we propose a procedure to deal with more general vectors than that generated by the
first eigenvector of the restriction of the Laplacian to the tube cross section. In simpler
situations, like tubes  in
$\R^3$ with vanishing ``twisting'' (see Definition~\ref{defTwist}) or in
$\R^2$, it is possible to get a suitable separation of variables, it is not necessary to employ that
procedure and the limit operator does not depend on the subspace considered. However, in the most
general case the limit operator will depend on the subspace in the cross section due to an
``additional memory of extra dimensions'' that can be present in $\R^3$. This will appear explicitly
in the expressions of effective potentials ahead.

\subsection{Tube and Hamiltonian}\label{subsecTH} Given a connected open set
$\Omega\subset\R^3$, denote by $-\Delta_\Omega$ the usual negative Laplacian operator with Dirichlet
boundary conditions, that is, the Friedrichs extensions of $-\Delta$ with domain
$\dom(-\Delta) = C_0^\infty(\Omega)$. More precisely, $-\Delta_\Omega$ is the self-adjoint operator
acting in
$\LL^2(\Omega)$ associated with the positive sesqui\-lin\-e\-ar form
\[ b_\Omega(\psi,\varphi)=\la \nabla_x\psi,\nabla_x \varphi\ra,\quad \dom b_\Omega =
\hil_0^1(\Omega);
\]the inner product is in the space $\LL^2(\Omega)$ and $\nabla_x$ is the usual gradient in
cartesian coordinates $(x,y,z)$. The operator $-\Delta_\Omega$ describes the energy of a quantum free
particle in $\Omega$.  We are interested in $\Omega$ representing the following kind of tubes, which
will be described in some details. Let
$\gamma:\R\to\R^3$ be a
$C^3$ curve parametrized by its (signed) arc length $s$, and introduce
\[ T(s)=\dot\gamma(s),\quad N(s)=\frac1{\kappa(s)}\dot T(s),\quad B(s)=T(s)\times N(s),
\]with $\kappa(s)=\|\ddot \gamma(s)\|$ being the curvature of~$\gamma$; the dot over a function always
indicates derivative with respect to~$s$. These quantities are the well-known tangent, normal and
binormal (orthonormal) vectors of~$\gamma$, respectively, which constitute a distinguished Frenet
frame for the curve
\cite{Kling}. If $\kappa$ vanishes in suitable intervals one considers a constant Frenet frame and in
many cases it is possible to join distinguished and constant Frenet frames, for instance if $\kappa>0$
on a bounded interval $I$ and vanishing in $\R\setminus I$ \cite{EKK}; this will be implicitly used
when we deal with the broken-line limit in Section~\ref{sectBrokenLine}. Here it is assumed that such
a global Frenet frame exists, and so it changes along the curve $\gamma$ according to the
Serret-Frenet equations
\[
\left( {\begin{array}{*{20}c}
   {\dot T}  \\
   {\dot N}  \\
   {\dot B}  \\
\end{array}} \right) = \left( {\begin{array}{*{20}c}
   0 & \kappa  & 0  \\
   { - \kappa } & 0 & \tau   \\
   0 & { - \tau } & 0  \\
\end{array}} \right)\left( {\begin{array}{*{20}c}
   T  \\
   N  \\
   B  \\
\end{array}} \right),
\]where $\tau(s)$ is the torsion of the curve~$\gamma(s)$. Although in principle closed curves can be
allowed we will exclude this case since our main interest is in curves of infinite length.
Nevertheless, the case of closed curves  would lead to a bounded tube and discrete spectrum of the
associated Laplacian $-\Delta^\varepsilon_\alpha$ (see below) and would fit in a small variation of
the discussion in Subsection~\ref{subsecBTubes}.  

Next an open, bounded and connected subset
$\emptyset\ne S\subset\R^2$ will be transversally linked to the reference curve
$\gamma$, so that $S$ will be the cross section of the tube. However, $S$ will also be rotated with
respect to the Frenet frame as one moves along the reference curve~$\gamma$, and with rotation angle
given by a $C^1$ function
$\alpha(s)$. Given $\varepsilon>0$, the tube so obtained (i.e., by moving $S$ along the curve $\gamma$
together with the rotation $\alpha(s)$) is given by
\[
\Omega_\alpha^\varepsilon := \{(x,y,z)\in\R^3: (x,y,z)=f^\varepsilon_\alpha(s,y_1,y_2), s\in\R,
(y_1,y_2)\in S  \},
\]with $f^\varepsilon_\alpha(s,y_1,y_2)=\gamma(s)+\varepsilon y_1 N_\alpha(s) + \varepsilon y_2
B_\alpha(s)$, and
\begin{eqnarray*} N_\alpha(s) &=& \cos\alpha(s) N(s) -\sin \alpha(s) B(s)
\\ B_\alpha(s) &=& \sin\alpha(s) N(s) +\cos\alpha(s) B(s).
\end{eqnarray*} The tube is then defined by the map $f_\alpha^\varepsilon:\R\times S\to
\Omega_\alpha^\varepsilon$, and we will be interested in the singular limit case
$\varepsilon\to0$, that is, when the tube is squeezed to the curve $\gamma$ and what happens to the
Dirichlet Laplacian $-\Delta_{\Omega_\alpha^\varepsilon}$ in this process (see \cite{FreitasKrej} for
the corresponding construction in $\R^n$). This will result in the one-dimensional quantum energy
operator that arises after the confinement onto
$\gamma$; on basis of Proposition~8.1 in \cite{FH}, it is expected that this will be the relevant
operator also in the case of holonomic constraints (at least from the dynamical point of view for
finite times); see also Theorem~3 in~\cite{DellAntTen}.  

Note that the tube is completely determined by the curvature
$\kappa(s)$ and torsion $\tau(s)$ of the curve
$\gamma(s)$, together with the cross-section $S$ and the rotation function~$\alpha(s)$.  Below some
conditions will be imposed on $f_\alpha^\varepsilon$ so that it becomes a
$C^1$-diffeomorphism. It will be assumed that $\gamma$ has no self-intersection and that its curvature
is a bounded function of~$s$, that is, $\|\kappa\|_\infty<\infty$, and, for simplicity, unless
explicitly specified that
$\|\tau\|_\infty,\|\dot\alpha\|_\infty<\infty$.    

As usual in this context we rescale and change variables in order to work with the fixed domain
$\R\times S$; the price we have to pay is a nontrivial Riemannian metric $G=G^\varepsilon_\alpha $,
which is induced by the embedding
$f_\alpha^\varepsilon$, that is, $G=(G_{ij})$, $G_{ij}=e_i\cdot e_j=G_{ji}$, $1\le i,j\le3$, with 
\[ e_1 = \frac{\partial f_\alpha^\varepsilon}{\partial s},\quad e_2 = \frac{\partial
f_\alpha^\varepsilon}{\partial y_1},\quad e_3 = \frac{\partial f_\alpha^\varepsilon}{\partial y_2}.
\]Direct computations give
\[ G_\alpha^\varepsilon =  \left( {\begin{array}{*{20}c}
   \beta_\varepsilon ^2 +\frac1{\varepsilon^2}(\rho_\varepsilon ^2+\sigma_\varepsilon ^2) &
\rho_\varepsilon   & \sigma_\varepsilon   \\
   { \rho_\varepsilon  } & \varepsilon^2 & 0   \\
   \sigma_\varepsilon  & { 0 } & \varepsilon^2  \\
\end{array}} \right),
\]with $\beta_\varepsilon(s,y_1,y_2)  = 1-\varepsilon \kappa(s)(y_1\cos\alpha(s) + y_2
\sin\alpha(s))$, $\rho_\varepsilon(s,y_1,y_2)  =-\varepsilon^2 y_2(\tau(s)-\dot\alpha(s))$ and
$\sigma_\varepsilon(s,y_1,y_2)  =\varepsilon^2 y_1(\tau(s)-\dot\alpha(s))$. Its determinant is 
\[ |\det G_\alpha^\varepsilon|=\varepsilon^4 \beta_\varepsilon ^2,
\]so that $f_\alpha^\varepsilon$ is a local diffeomorphism provided $\beta_\varepsilon $ does not
vanish on
$\R\times S$, which will occur if $\kappa$ is bounded and $\varepsilon$ small enough (recall that $S$
is a bounded set), so that $\beta_\varepsilon >0$. By requiring that
$f_\alpha^\varepsilon$ is injective (that is, the tube is not self-intersecting) one gets a global
diffeomorphism.     

Coming back to the beginning of this subsection, we pass the sesqui\-lin\-e\-ar
form in usual coordinates $(x,y,z)$,
\[ b_{\Omega^\varepsilon_\alpha}(\psi,\varphi)=\la \nabla_x \psi,\nabla_x \varphi\ra,\quad
\dom b_{\Omega^\varepsilon_\alpha} = \hil_0^1(\Omega^{\varepsilon}_\alpha),
\] to coordinates $(s,y_1,y_2)$ of $\R\times S$ and express the inner product and gradients
appropriately. For
$\psi\in\hil^1_0(\Omega_\alpha^\varepsilon)$ set $\psi(s,y_1,y_2) :=
\psi(f_\alpha^\varepsilon(s,y_1,y_2))$. If  $\nabla$ denotes the gradient in the 
$(s,y_1,y_2)$ coordinates, then by the chain rule $\nabla_x\psi = J^{-1}\nabla\psi$ where $J$  is the
$3\times3$ matrix, expressed in the Frenet frame $(T,N,B)$,
\begin{eqnarray*} J&=&\left( {\begin{array}{*{20}c}
   e_1  \\   { e_2 }    \\   e_3   \\
\end{array}} \right) \\ &=& 
\left({\begin{array}{*{20}c}
   \beta_\varepsilon  & \varepsilon (\tau-\dot\alpha) (y_1\sin\alpha-y_2\cos\alpha) &
\varepsilon (\tau-\dot\alpha) (y_2\sin\alpha+y_1\cos\alpha)  \\
  0 & -\varepsilon\cos\alpha & \varepsilon\sin\alpha  \\
   0 & \varepsilon\sin\alpha &\varepsilon\cos\alpha  \\
\end{array}} \right).
\end{eqnarray*} Noting that $JJ^t=G$, $\det J =|\det G|^{1/2}=\varepsilon^2 \beta_\varepsilon
$, and introducing the notation 
\[
\la \psi,\varphi\ra_G = \int_{\R\times S}\overline{\psi(s,y_1,y_2)} \varphi(s,y_1,y_2)\;
\varepsilon^2 \beta_\varepsilon (s)\;dsdy_1dy_2,
\]it follows that
\[
\la \nabla_x \psi,\nabla_x \varphi\ra= \la J^{-1}\nabla \psi, J^{-1}\nabla \varphi\ra_G=\la
\nabla \psi, G^{-1}\nabla \varphi\ra_G
\] and the operator $-\Delta^\varepsilon_\alpha$ can be described as the operator associated with the
positive sesqui\-lin\-e\-ar form 
\[
\dom \tilde b^\varepsilon = \hil_0^1(\R\times S,G),\quad \tilde b^\varepsilon(\psi,\phi) :=
\la \nabla \psi, G^{-1}\nabla \varphi\ra_G,
\]and the $\varepsilon$ dependence is now in the Riemannian metric $G$. More precisely, the above
change of variables is implemented by the unitary transformation 
\[ U:\LL^2(\Omega_\alpha^\varepsilon)\to \LL^2(\R\times S,G),\quad U\psi = \psi\circ
f_\alpha^\varepsilon. 
\] Note, however, that usually we will continue denoting $U\psi$ simply by~$\psi$. Explicitly the
quadratic form is given by (with $dy=dy_1dy_2$ and $\nabla_\perp \psi =
(\partial_{y_1}\psi,\partial_{y_2}\psi)$, so that $\nabla = (\partial_s,\nabla_\perp)$)
\begin{eqnarray*}
\tilde b^\varepsilon(\psi) &=& \left\| J^{-1}\nabla \psi  \right\|^2_G\\
&=&\varepsilon^2\int_{\R\times S} dsdy 
\left[\frac1{\beta_\varepsilon }\left|\nabla \psi
\cdot(1,y_2(\tau-\dot\alpha),y_1(\tau-\dot\alpha)) \right|^2 + \frac{\beta_\varepsilon
}{\varepsilon^2}|\nabla_\perp
\psi|^2\right] \\ &=&\varepsilon^2\int_{\R\times S} dsdy 
\left[\frac1{\beta_\varepsilon }\left|\nabla \psi \cdot(1,Ry(\tau-\dot\alpha)) \right|^2 +
\frac{\beta_\varepsilon }{\varepsilon^2}|\nabla_\perp \psi|^2\right] ,
\end{eqnarray*} where $R=\left( {\begin{array}{*{20}c}
   0 & 1  \\   1 & 0  \\
\end{array}} \right)$.
 On functions $\psi\in C_0^\infty(\R\times S)$ a  calculation shows that the corresponding operator
has the following action
\[ U (-\Delta_\alpha^\varepsilon)U^* \psi = -\varepsilon^2 \frac1{\beta_\varepsilon }
\mathrm{div} \beta_\varepsilon  G^{-1}\nabla\psi, 
\]and mixed derivatives will not be present iff $\tau(s)-\dot\alpha(s)=0$, since this is the condition
for
$G^{-1}$ be a diagonal matrix:
\[ G^{-1} = \left({\begin{array}{*{20}c}
   \beta_\varepsilon ^{-2} & 0 & 0 \\
  0 & \varepsilon^{-2} & 0  \\
   0 & 0 &\varepsilon^{-2}  \\
\end{array}} \right).
\]This hypothesis simplifies the operator expression and it is the main reason this case attracted
more attention
\cite{DE,ClaBra}; more recently the general case has also been considered and the spectral-geometric
effects have been reviewed in \cite{Krej}. With respect to the limit $\varepsilon\to0$ for the more
general tubes we consider here, the expression of the operator actions seem prohibitive to be
manipulated in a useful way; so the main advantage  in using sesqui\-lin\-e\-ar forms and
$\Gamma$-con\-ver\-gence, as studied in case of bounded tubes in~\cite{BMT}. 
 
Since we are interested in the limit of the tube approaching the reference curve~$\gamma$, i.e.,
$\varepsilon\to0$, it is still necessary to perform two kinds of ``regularizations'' in $\tilde
b^\varepsilon$ in order to extract a meaningful limit; these are common approaches to balance singular
problems, particularly due to the presence of regions that scale in different manners, and so to put
them in a tractable form \cite{Braides}. The first one is physically related to the uncertainty
principle in quantum mechanics  and is in fact a renormalization. Let $u_n\in \hil_0^1(S)$ and
$\lambda_n\in\R$, $n\ge0$, be the normalized eigenfunctions and corresponding eigenvalues of the
(negative) Laplacian restricted to the cross section~$S$ (since $S$ is bounded the Dirichlet Laplacian
on~$S$ has compact resolvent), and we suppose that
$\lambda_0<\lambda_1<\lambda_2\cdots$ and that all eigenvalues
$\lambda_n$ are simple; later on we will underline where this assumption is used (see, in particular, Subsection~\ref{subsecFS}). 

When the tube is squeezed there are divergent energies due to terms of the form
$\lambda_n/\varepsilon^2$, and one needs a rule to get rid of these energies related to transverse
oscillations in the tube. Since in quantum mechanics these quadratic forms~$\tilde b^\varepsilon$
correspond to expectation values of total energy, one subtracts such diverging terms from~$\tilde
b^\varepsilon$. However, in principle it is not clear which expression one should use in the
subtraction process and we will consider the following possibilities
\[
\tilde b^\varepsilon(\psi)- \frac{\lambda_n}{\varepsilon^2}\|\psi\|^2_G =\tilde b^\varepsilon(\psi)-
{\lambda_n}\int_{\R\times S}dsdy\; \beta_\varepsilon(s,y) |\psi(s,y)|^2.
\]  

The second regularization is simply a division by the global factor~$\varepsilon^2$ so defining the
family of quadratic forms we will work with:
\begin{eqnarray*} b^\varepsilon_n(\psi) &:=& \varepsilon^{-2}\left( \tilde b^\varepsilon(\psi)-
\frac{\lambda_n}{\varepsilon^2}\|\psi\|^2_G  \right) 
\\ &=& \int_{\R\times S} dsdy  \left[\frac1{\beta_\varepsilon }\left|\nabla \psi
\cdot(1,Ry(\tau-\dot\alpha)) \right|^2 + \frac{\beta_\varepsilon }{\varepsilon^2}\left(|\nabla_\perp
\psi|^2-\lambda_n|\psi|^2\right)\right],
\end{eqnarray*}
$\dom b^\varepsilon_n = \hil^1_0(\R\times S)$. In \cite{BMT} only $b_0^\varepsilon$ was considered,
and ahead we propose a procedure to deal with~$b^\varepsilon_n$, for all $ n\ge0$. After such
regularizations we finally obtain the operators $H^\varepsilon_n$, associated with these forms, for
which  we will investigate the limit
$\varepsilon\to0$.  Let
$\hil^\varepsilon$ denote the Hilbert space $\LL^2(\R\times S,\beta_\varepsilon)$, that is, the inner
product is given by
\[
\la\psi,\varphi\ra_{\varepsilon} := \int_{\R\times S}
\overline{\psi(s,y)}\varphi(x,y)\,\beta_\varepsilon(s,y)\,dsdy.
\]

\begin{Definition}\label{defHneps}
 The operator $H_n^\varepsilon$ is the self-adjoint operator associated with the sesqui\-lin\-e\-ar
form
$b_n^\varepsilon$ (see \cite{ISTQD}, page~101), whose domain $\dom H_n^\varepsilon$ is dense in $\dom
b_n^\varepsilon$ and
\[ b_n^\varepsilon(\psi,\varphi) = \la \psi,H^\varepsilon_n\varphi\ra_\varepsilon,\quad
\forall \psi\in\dom b^\varepsilon_n,\forall \varphi\in \dom H^\varepsilon_n. 
\] In case the functions $\kappa,\tau,\dot\alpha$ are bounded and~$S$ has a smooth boundary then, by elliptic regularity \cite{agmon}, $\dom H_n^\varepsilon =\hil^2(\R\times S)\cap
\hil^1_0(\R\times S)$.
\end{Definition}

\begin{Remark} In the above construction of $\hil^\varepsilon$ and $H^\varepsilon_n$ it was very
important the uniform  bound of $\beta_\varepsilon$, which implies that for all $\varepsilon>0$ small
enough there exist
$0<a_\varepsilon\le a^\varepsilon<\infty$ with
\[ a_\varepsilon \|\cdot\|\le \|\cdot\|_\varepsilon \le a^\varepsilon\|\cdot\|,
\] (similar inequalities also hold for the associated quadratic forms $b_n^\varepsilon$) so that all
Hilbert spaces
$\hil^\varepsilon$ coincide algebraically with
$\LL^2(\R\times S)$ and also have equivalent norms. Furthermore, it is possible to assume that both
$a_\varepsilon\to1$ and
$a^\varepsilon\to1$ hold as
$\varepsilon\to0$, since the sequence of functions $\beta_\varepsilon\to1$ uniformly. These properties
imply the equality of the quadratic form domains for all $\varepsilon>0$ and permit us to speak of
$\Gamma$-con\-ver\-gence of
$b^\varepsilon_n$ and resolvent con\-ver\-gence of
$H_n^\varepsilon$ in $\LL^2(\R\times S)$ even though they act in, strictly speaking, different Hilbert
spaces. Such facts will be freely used ahead.
\end{Remark}

\begin{Definition}\label{defTwist} Given $\Omega_\alpha^\varepsilon$, let $C_n(S) :=\int_S dy
\left|\nabla_\perp u_n(y)\cdot R y\right|^2$, where, as before, $u_n$ is the $(n+1)$th normalized
eigenfunction of the negative Laplacian on~$S$. The tube
$\Omega_\alpha^\varepsilon$ is said to be {\em quantum twisted} if for some
$n$ the function
\[
\mathcal A_n(s) := (\tau(s)-\dot\alpha(s))^2 C_n(S)\, \ne\,0.
\]
\end{Definition}

Note that the quantities $C_n(S)$ are  parameters that depend only on the cross section~$S$, and that 
$\mathcal A_n$ is zero if either the cross section rotation compensates the torsion of the curve (i.e.,
$\tau-\dot\alpha=0$) or if the eigenfunction $u_n$ is radial (i.e.,
$C_n(S)=0$); $\mathcal A_n$ has a geometrical nature and $\mathcal A_0$ was first considered
\cite{BMT}.  In \cite{Krej} there is an interesting list of equivalent formulations of the condition
$\tau-\dot\alpha=0$. 

It should be mentioned that the case of twisting has been addressed also in the case of infinite
curve~$\gamma$  in 
\cite{EKK}, and that there are some studies realized in $\R^n$ for $n\ge3$ \cite{FreitasKrej}.

\subsection{Confinement: Statements} In this subsection we discuss results about the limit
$\varepsilon\to0$ of  $H^\varepsilon_n$. We begin with $H^\varepsilon_0$ and will follow closely
\cite{BMT}, where
$\Gamma$-con\-ver\-gence was used to study the limit operator upon confinement and con\-ver\-gence of
eigenvalues in case of  bounded curves~$\gamma$ (so bounded tubes). We will
point out the necessary modifications in their approach in order to prove Theorem~\ref{teorHzero}
below, which is in the setting of 
$\LL^2(\R\times S)$ and no restriction on the curve length, although our results on spectral
con\-ver\-gence are limited to the standard consequences of strong operator resolvent con\-ver\-gence
(not stated here; see, for instance, Corollary~10.2.2 in~\cite{ISTQD}); in any event, ahead we shall
recover their spectral con\-ver\-gence by means of our Proposition~\ref{GammaNorm}.  

For each
$n\ge0$, denote by $\LL_n$ (resp.\ $h_n$) the Hilbert subspace of $\LL^2(\R\times S)$ (resp.\
$\hil^1_0(\R\times S))$ of vectors of the form $\psi(s,y)=w(s)u_n(y)$, $w\in \LL^2(\R)$ (resp.\ $w\in
\hil^1(\R)=\hil_0^1(\R)$), so that 
\begin{eqnarray*}
\LL^2(\R\times S) &=& \LL_0\oplus \LL_1\oplus \LL_2\oplus\cdots, \\ \hil^1_0(\R\times S) &=& h_0\oplus
h_1\oplus h_2\oplus\cdots,
\end{eqnarray*} and each $h_n$ is a dense subspace of~$\LL_n$. Note that $h_n$ is related to
$b^\varepsilon_n$. In Theorem~\ref{teorHzero} we deal with the limit of $b^\varepsilon_0$ and
generalize some results of \cite{BMT} for elements of $h_0$; then we propose a procedure to deal with
the relation between the limit of
$b_n^\varepsilon$ and~$h_n$, for all $ n$.   

Now define the real potentials
\begin{eqnarray*} V_n^{\mathrm{eff}}(s) &:=& \mathcal A_n(s) - \frac14 \kappa(s)^2 \\ &=&
(\tau(s)-\dot\alpha(s))^2\,C_n - \frac14 \kappa(s)^2,\quad n\ge0,
\end{eqnarray*} and the corresponding Schr\"odinger operators
\[ (H_n^0\psi)(s) = -\frac{d^2}{ds^2}\psi(s) + V_n^{\mathrm{eff}}(s)\psi(s),
\]whose domains are $\dom H^0_n = \hil^2(\R),$ for all~$n\ge0,$ in case of bounded functions
$\kappa,\tau,\dot\alpha$. Here the curvature $\kappa$ is always assumed to be bounded, but if either
the torsion
$\tau$ or the $\dot\alpha$ is not bounded, the domain must be discussed on an almost case-by-case
basis. The subspace  $\LL_n$ can be identified with $\overline{\dom H^0_n}=\LL^2(\R)$ via
$wu_n\mapsto w$. Let $b_n^0$ be the sesqui\-lin\-e\-ar form generated by
$H_n^0$, that is, 
\[ b^0_n(\psi) = \int_\R ds\;\left( |\dot \psi(s)|^2 + V_n^{\mathrm{eff}}(s)|\psi(s)|^2\right)
\] whose $\dom b^0_n=\hil^1(\R)$ (for bounded $\kappa,\tau,\dot\alpha$) can be identified with
$h_n$; hence $b_n^0(\psi)=\infty$ if $\psi\in \LL_n\setminus h_n$. 

\begin{Theorem}\label{teorHzero} The sequence of self-adjoint operators $H^\varepsilon_0$ converges in
the strong resolvent sense  to~$H_0^0$ in $L_0$ as $\varepsilon\to0$.
\end{Theorem}

\begin{Remark} We have then obtained an explicit form of the effective operator that describes the
confinement of a free quantum particle (in the subspace $\LL_0$) from
$\Omega_\alpha^\varepsilon$ to the curve $\gamma$. The action of the operator $H_0^0$ is the same as
that in~\cite{BMT} and, as already anticipated, the arguments below for its proof are based on
$\Gamma$-con\-ver\-gence of quadratic forms and mainly consist of indications of the necessary
modifications to take into account unbounded curves, that is,  the results of
Section~\ref{sectionGammaConv} and the important additional verification of weak
$\Gamma$-con\-ver\-gence of the involved quadratic forms.
\end{Remark}

It is possible to get some intuition about what should be expected for the con\-ver\-gence of
$H^\varepsilon_n$, $n\ge1$, by considering the form $b^\varepsilon_n$ evaluated at vectors of~$h_n$,
i.e., 
$wu_n$, $w\in \hil^1(\R)$, and some formal arguments.  A direct substitution gives an integral with
four terms
\begin{eqnarray*} b_n^\varepsilon(wu_n) &=& \int_{\R\times S} dsdy\, \Big[
\frac1{\beta_\varepsilon} |\dot w|^2  |u_n|^2 + |w|^2 \Big(
\frac1{\beta_\varepsilon}|\nabla_\perp u_n\cdot Ry(\tau-\dot\alpha)|^2 \\ &+&
\frac{\beta_\varepsilon}{\varepsilon^2}\left( |\nabla_\perp u|^2 - \lambda_0|u|^2 \right)
\Big)+  \frac1{\beta_\varepsilon}  2\Ree (\overline{\dot w u_n}\, w \nabla_\perp u_n\cdot Ry)
\Big],
\end{eqnarray*} and if we approximate $\beta_\varepsilon\approx 1$, except in the third term, we find
that the last term vanishes due to the Dirichlet boundary condition and thus (recall that $u_n$ is
normalized in
$\LL^2(S)$)
\begin{eqnarray*} b^\varepsilon_n(wu_n) &\approx& \int_\R ds \left[|\dot w(s)|^2 +  C_n(S)
(\tau(s)-\dot\alpha(s))^2 |w(s)|^2\right] \\ &+& \int_{\R} ds\, |w(s)|^2 \int_S
dy\,\frac{\beta_\varepsilon}{\varepsilon^2}\left( |\nabla_\perp u_n(y)|^2 -
\lambda_n|u_n(y)|^2 \right) \big),
\end{eqnarray*}   so that the first term $\mathcal A_n(s)$ in the effective potential
$V_n^{\mathrm{eff}}(s)$ is clearly visible. However, the remaining term
\[ K_n^\varepsilon(u,s):=\int_S dy\,\frac{\beta_\varepsilon(s,y)}{\varepsilon^2}\left( |\nabla_\perp
u(y)|^2 -
\lambda_n|u(y)|^2 \right) 
\] is  related to the curvature and requires much more work; this was a major contribution of
\cite{BMT} in case $n=0$ through the study of minima and strong $\Gamma$-con\-ver\-gence.   It is
shown in
\cite{BMT} that $K_0^\varepsilon$ attains a minimum $-\kappa^2/4>-\infty$ since
$\lambda_0$ is the bottom of the spectrum of the Laplacian restricted to the cross section~$S$;
however, a minimum is not expected to occur in case of $K_n^\varepsilon$, $n\ge1$, since, if we take
again the approximation
$\beta_\varepsilon\approx 1$, there will be vectors
$u\in h_{j}$, $j<n$, such that 
\[ K_n^\varepsilon(u,s)\approx \frac1{\varepsilon^2} (\lambda_{j}-\lambda_n)\to -\infty,\quad
\varepsilon\to0,
\]  and this unboundedness from below pushes~$u$ away of the natural range of applicability of
quadratic forms and
$\Gamma$-con\-ver\-gence. We then impose that if  
$b^\varepsilon_n(\psi)\to\pm \infty$, then the vector~$\psi$ is simply excluded from the domain of the
limit sesqui\-lin\-e\-ar form, and so also from the domain of the  limit operator.  Since
Theorem~\ref{teorHzero} tells us how to deal with vectors in $\LL_0$, by taking the above discussion
into account we propose to study the con\-ver\-gence of
$b^\varepsilon_1$ restricted to the orthogonal complement of $\LL_0$ in $\LL^2(\R\times S)$ and, more
generally, to study the $\Gamma$-con\-ver\-gence of $b^\varepsilon_n$ restricted to 
\[
\mathcal E_{n}:=\left(\LL_0\oplus \LL_1\oplus\cdots \oplus\LL_{n-1}\right)^\perp,\quad n\ge1.
\] In what follows these domain restrictions will also apply to the resolvent con\-ver\-gence of the
restriction to~$\mathcal E_n$ of the associated self-adjoint operators and will be used without
additional warnings. Under such procedure we shall obtain the following result (recall that  it is
assumed that all eigenvalues $\lambda_n$ of the Laplacian restricted to the cross section~$S$ are
simple).
 
\begin{Theorem}\label{teorHn}  The sequence of self-adjoint operators obtained by the restriction of
$H^\varepsilon_n$ to $\dom H_n^\varepsilon\cap\mathcal E_n$ converges, in the strong resolvent sense in
$L_n$, to~$H_n^0$ as $\varepsilon\to0$.
\end{Theorem}

\begin{Remark} Since quadratic forms can conveniently take the value $+\infty$, here we have rather
different nomenclatures for singular con\-ver\-gences, that is, the sequence of forms
$b^\varepsilon_n$ converges in~$\mathcal E_n$ whereas the matching sequence of operators
$H_n^\varepsilon$ converges in~$\LL_n$.
\end{Remark}

\begin{Remark} Note that this procedure leads to a kind of decoupling among the subspaces~$h_n$, which
is supported by some cases of plane waveguides treated in~\cite{ACF} where the decoupling is found
through an explicitly separation of variables; see the presence of Kronecker deltas in the expressions
for resolvents in Theorem~1 and Lemma~3 in~\cite{ACF}. In 3D such separation occurs if there is no
twisting~\cite{DE}.
\end{Remark}

\begin{Remark}\label{remarkResonances}
 In the first approximation, in the vicinity of the overall spectral threshold (i.e., by considering $\lambda_0$), the effective Hamiltonian describes bound states. However, at higher thresholds (i.e., $\lambda_n,n\ge1$), the effective Hamiltonians are usually related to resonances in thin tubes (see \cite{DEM,DES}), and so this physical content is another motivation for the consideration of  $b_n, n\ge1$. 
\end{Remark}

\begin{Remark} A novelty with respect to similar situations in $\R^2$ \cite{ACF,DellAntTen} is that
the effective operator~$H_n^0$ describing the particle confined to the curve $\gamma(s)$ depends
on~$n$, that is, we have different quantum dynamics for different subspaces~$L_n$. As mentioned in the
Introduction, this is a counterpart of  some situations found in
\cite{Takens,FH}, since there the effective dynamics may depend on the initial condition.
\end{Remark}

Since the above procedure was strongly based on the mentioned estimate
\[
K_n^\varepsilon(u,s)\approx \frac1{\varepsilon^2} (\lambda_{j}-\lambda_n)\to -\infty,\quad j<n,
\] as well as the rather natural expectation that it is ``spontaneously'' possible an energy
transfer only from higher to lower energy states, we present a rigorous version of such estimate. 
\begin{Proposition}\label{propLeakHL} For normalized $u(s,y)=w(s)u_j(y)$, $w\in\hil^1(\R)$, one has 
\[
\lim_{\varepsilon\to0} \varepsilon^2 b^\varepsilon_n(wu_j) = (\lambda_j-\lambda_n).
\]In particular, for small $\varepsilon$ and a.e.~$s$, $K_n^\varepsilon(u,s)\approx
\frac1{\varepsilon^2} (\lambda_{j}-\lambda_n)\to -\infty$ if $j<n$ as $\varepsilon\to0$.
\end{Proposition}
\begin{proof} Again a direct substitution and taking into account the Dirichlet boundary condition
yield
\begin{eqnarray*} b_n^\varepsilon(wu_j) &=& \int_{\R\times S} dsdy\, \Big[
\frac1{\beta_\varepsilon} |\dot w|^2  |u_j|^2 \\ &+& |w|^2 \Big(
\frac1{\beta_\varepsilon}|\nabla_\perp u_j\cdot Ry(\tau-\dot\alpha)|^2  +
\frac{\beta_\varepsilon}{\varepsilon^2}\left( |\nabla_\perp u_j|^2 - \lambda_0|u_j|^2 \right)
\Big)\Big],
\end{eqnarray*} and the only term that may not vanish as $\varepsilon\to0$ in $\varepsilon^2
b_n^\varepsilon(wu_j)$ is the last one, so it is enough to analyze 
\[ K_n^\varepsilon(u_j,s):=\int_S dy\,\frac{\beta_\varepsilon(s,y)}{\varepsilon^2}\left( |\nabla_\perp
u_j(y)|^2 -
\lambda_n|u_j(y)|^2 \right).
\]

Write $\beta_\varepsilon = 1- \xi\cdot y$, with $\xi = \varepsilon \kappa(s) z_\alpha$ and
$z_\alpha=(\cos\alpha, \sin\alpha)$. Since $u_j\in\hil^2(S)$ and $u_j$ satisfies the Dirichlet
boundary condition, upon integrating by parts one gets, for a.e.~$s$,
\begin{eqnarray*} K_n^\varepsilon(u_j,s) &=& -\frac{1}{\varepsilon^2} \int_Sdy\,
\overline{u_j}\left( \nabla_\perp\cdot(\beta_\varepsilon \nabla_\perp u_j) 
+\lambda_n\beta_\varepsilon u_j)\right)
\\ &=&
 -\frac{1}{\varepsilon^2} \int_Sdy\, \overline{u_j}\left( \nabla_\perp\beta_\varepsilon\cdot
\nabla_\perp u_j + \beta_\varepsilon(\Delta_\perp u_j +\lambda_n u_j\right)) \\ &=&
  -\frac{1}{\varepsilon^2} \int_Sdy\, \overline{u_j}\left(-\xi\cdot \nabla_\perp u_j +
\beta_\varepsilon(-\lambda_j+\lambda_n) u_j\right)  \\ &=&
  \frac{1}{\varepsilon} \kappa(s)z_\alpha\cdot \int_Sdy\, \overline{u_j}\nabla_\perp u_j +
\frac{(\lambda_j-\lambda_n)}{\varepsilon^2} \int_S dy\, \beta_\varepsilon |u_j|^2.
\end{eqnarray*} Exchanging the roles of $u_j$ and $\overline{u_j}$ yields
\[ K_n^\varepsilon(u_j,s) = \frac{1}{\varepsilon} \kappa(s)z_\alpha\cdot
\int_Sdy\,{u_j}\nabla_\perp \overline{u_j} +\frac{(\lambda_j-\lambda_n)}{\varepsilon^2} \int_S dy\,
\beta_\varepsilon |u_j|^2,
\]and by adding such expressions 
\[ K_n^\varepsilon(u_j,s) = \frac{1}{2\varepsilon} \kappa(s)z_\alpha\cdot
\int_Sdy\,\nabla_\perp|u_j|^2 + \frac{(\lambda_j-\lambda_n)}{\varepsilon^2} \int_S dy\,
\beta_\varepsilon |u_j|^2.
\]
 The Dirichlet boundary condition implies $\int_Sdy\,\nabla_\perp|u_j|^2=0$. Since by dominated
con\-ver\-gence
$\int_S dy\, \beta_\varepsilon |u_j|^2\to \|u_j\|^2=1$ as
$\varepsilon\to0$, it follows that $\varepsilon^2 K_{n,j}(s)\to (\lambda_j-\lambda_n)$ for
a.e.~$s\in\R$. Again by dominated con\-ver\-gence
\[
\lim_{\varepsilon\to0} \varepsilon^2b^\varepsilon_n(wu_j) = (\lambda_j-\lambda_n).
\] The proof is complete.
\end{proof}

The reader may protest at this point that the above procedure of considering restriction to $\mathcal
E_n$ when dealing with $b_n$, $n>0$, should in fact be deduced instead of assumed. However, we think
that such deduction is certainly  beyond the range of
$\Gamma$-con\-ver\-gence, since by  beginning with the form $b^\varepsilon_n$ the vectors $\psi\in
h_j$, $j>n$, will not belong to the domain of the limit form $b^0_n$ since
$\lim_{\varepsilon\to0}b_n^\varepsilon(\psi)=+\infty$, while for vectors $\psi\in h_j$, $j<n$, the 
$\lim_{\varepsilon\to0}b_n^\varepsilon(\psi)=-\infty$ indicates that they would be outside the usual
scope of limit of forms and $\Gamma$-con\-ver\-gence. Further, we have got reasonable expectations that
the domain of the limit form
$b_n^0$ should consists of only vectors in $h_n\subset \LL_n$. By taking into account the proposed
procedure of suitable domain restrictions,  in the next subsection we will support such expectations
by proving the above theorems; we stress that there is no mathematical issue  in the formulation of
the above procedure (e.g., if
$\psi\in\dom H_n^\varepsilon\cap {\mathcal E}_n$, then $H_n^\varepsilon\psi\in {\mathcal E}_n$) and
that it rests only on physical interpretations.

\subsection{Confinement: Convergence} In this subsection the proofs of Theorems~\ref{teorHzero}
and~\ref{teorHn} are presented. There are two main steps; the first one is to check the strong
$\Gamma$-con\-ver\-gence (recall the restriction to the subspace
$\mathcal E_n$)
\[ b^\varepsilon_n\sgconv b^0_n,
\]and the second one is a verification that if $\psi_\varepsilon\wseta \psi$ in $\mathcal E_n\subset
\LL^2(\R\times S)$, then
\[ b^0_n(\psi)\le \liminf_{\varepsilon\to0} b_n^\varepsilon(\psi_\varepsilon). 
\] The theorems will then follow by Theorem~\ref{mainTheorGamma}, including the additional information
$b^\varepsilon_n\wgconv b^0_n$ in~$\mathcal E_n$. The proof of the first step is just a verification
that the proof of the corresponding result in~\cite{BMT} can  be explicitly carried out for unbounded tubes and for $b_n^\varepsilon$ with $n\ge0$. Then we shall address the proof of the second step and at the end of this subsection we
briefly discuss how the spectral con\-ver\-gences in~\cite{BMT} can be recovered from our
Proposition~\ref{GammaNorm}. 
 
\subsubsection{First Step}\label{subsecFS} I will be very objective and just indicate the adaptations to the
proofs in~\cite{BMT}; I will also use a notation similar to the one used in that work. The first point
to be considered is a variation of Proposition~4.1 in~\cite{BMT}, that is, the study of the quantities 
\[
\lambda_n(\xi) := \inf_{{}^{\quad0\ne v\in\hil_0^1(S)}_{v\in [u_0,\cdots,u_{n-1}]^\perp}}
\frac{\int_S (1-\xi\cdot y) |\nabla_\perp v(y)|^2}{\int_S (1-\xi\cdot y) |v(y)|^2},\quad n>0,
\]and 
\[
\gamma_{\varepsilon,n}(s) := \frac1{\varepsilon^2} \left( \lambda_n(\xi)-\lambda_n 
\right),\quad \xi=\varepsilon \kappa(s) z_\alpha(s);
\] recall that $z_\alpha=(\cos\alpha,\sin\alpha)$. In case $n=0$ the above
infimum is taken over $0\ne v\in\hil_0^1(S)$; $[u_0,\cdots,u_n]$ denotes the subspace in
$\hil_0^1(S)$ spanned by the first eigenvectors $u_0,\cdots,u_n$ of the (negative) Laplacian
restricted to the cross section~$S$ and recall that here its eigenvalues
$\lambda_0<\lambda_1<\lambda_2\cdots$ are supposed to be simple.

The prime goal is to show that 
\[
\gamma_{\varepsilon,n}(s) \to -\frac14 \kappa^2(s)
\]uniformly on~$\R$ (Proposition~4.1 in~\cite{BMT}; don't confuse $\gamma_{\varepsilon,n}(s)$ with the
reference curve~$\gamma(s)$). In order to prove this we consider the unique solution
$u_{\xi,n}\in\hil_0^1(S)$ of the problem
\[ -\Delta_\perp u_{\xi,n}-\lambda_nu_{\xi,n} = -\xi\cdot \nabla_\perp u_{n}, \quad u_{\xi,n}\in
[u_0,\cdots,u_n]^\perp,
\]which exists by Fredholm alternative since $[u_0,\cdots,u_n]^\perp$ is invariant under the operator
$(-\Delta_\perp)$, the real number $\lambda_n$ belongs to the resolvent set of the restriction
$(-\Delta_\perp)|_{[u_0,\cdots,u_n]^\perp}$ (again because the eigenvalues are simple) and $\xi\cdot
\nabla_\perp u_n$ is orthogonal to $u_n$. Note that the simplicity of eigenvalues was again invoked to guarantee the unicity of solution
and that
$u_{\xi,n}$ is real  since its complex conjugate is also a solution of
the same problem.   

Since we can take real-valued eigenfunctions $u_n$,  the 
Lemma~4.3 in~\cite{BMT}  will read
\[
\inf_{{}^{\quad v\in\hil_0^1(S)}_{v\in[u_0,\cdots,u_{n-1}]^\perp}} \int_S dy\; \left[ |\nabla_\perp
v|^2 -
\lambda_n|v|^2 + (\xi\cdot \nabla_\perp{u_n})(v+\bar v) \right]= -\frac14
\xi^2,
\]with the infimum reached precisely for $u_{\xi,n}$, also a real-valued function. With such remarks
the proof of Lemma~4.3 in~\cite{BMT} also works for all~$n\ge0$. The same
words can be repeated for the proof of the above mentioned Proposition~4.1, the only necessary
additional remark is that the function $\varphi$ in their equation~(4.12) can also be taken real,
since in our context it always appears in the form
$(\varphi+\bar\varphi)$. With such modifications we obtain suitable versions of that proposition, that
is, 
$\gamma_{\varepsilon,n}(s) \to -\frac14 \kappa^2(s)$ uniformly on~$\R$ as $\varepsilon\to0$ for
all~$n\ge0$.   With such tools at hand the proof of
$b^\varepsilon_n\sgconv b^0_n$ in~$\mathcal E_{n}$, for all $n\ge0$, is exactly that done in
Section~4.3 of~\cite{BMT}. This completes our first main step. 
\subsubsection{Second Step}\label{subsubsectSS} Now we are going to complete the proofs of
Theorems~\ref{teorHzero} and~\ref{teorHn}. However we consider a simple modification of the forms
$b^\varepsilon_n$, that is, we shall consider
\[ b^\varepsilon_{n,c}(\psi) :=b^\varepsilon_{n}(\psi) + c\|\psi\|_\varepsilon^2,\quad \dom
b^\varepsilon_{n,c}=\dom b^\varepsilon_{n}\subset \mathcal E_n, 
\]for some $c\ge\|\kappa\|^2_\infty$ so that we work with the more natural set of positive forms. This
implies that if
$H^\varepsilon_n$ is the self-adjoint operator associated with $b^\varepsilon_n$, then
$H^\varepsilon_{n,c}:=H^\varepsilon_n+c\Id$ is the operator associated with
$b^\varepsilon_{n,c}$ and vice versa; so strong resolvent con\-ver\-gence $H^\varepsilon_n\to H^0_n$
in $\LL_n$ is equivalent to strong resolvent con\-ver\-gence $H^\varepsilon_{n,c}\to
H^0_{n,c}:=H^0_n+c\Id$ in~$\LL_n$ as~$\varepsilon\to0$.   

For each $\varepsilon>0$ and
$\psi\in \dom b^\varepsilon_{n,c}$ we have the important lower bound (by the definitions of the forms
and
$\gamma_{\varepsilon,n}$; see also equation~4.19 in~\cite{BMT})
\[ b^\varepsilon_{n,c}(\psi) \ge \int_{\R\times S}dsdy \left( \frac1{\beta_\varepsilon} \left|
\dot\psi + \nabla_\perp \psi\cdot Ry(\tau-\dot\alpha) \right|^2 + \beta_\varepsilon
(c-\gamma_{\varepsilon,n})|\psi|^2\right).
\] Our interest in considering the modified forms $b^\varepsilon_{n,c}$ is also that, for $\varepsilon$
small enough, 
\[ b^\varepsilon_{n,c}(\psi) \ge \left(c-\frac12 \|\kappa\|^2_\infty\right)\|\psi\|^2\ge
\frac12 \|\kappa\|^2_\infty\,\|\psi\|^2
\] and so both $b^\varepsilon_{n,c}$ and $H_{n,c}^\varepsilon$ are positive; this will be important
when dealing with weak con\-ver\-gence ahead. Furthermore, the arguments in First Step above and
Theorem~\ref{GammaPerturb2}, with $f$ playing the role of the norm, which is strongly continuous, also
show that
\[
 b^\varepsilon_{n,c}\sgconv b^0_{n,c}\quad{\mathrm{in}}\quad \mathcal E_n,
\]where 
\[
 b^0_{n,c}(\psi) := \int_\R ds\;\left( |\dot \psi(s)|^2 + \left[(\tau(s)-\dot\alpha(s))^2\,C_n +c-
\frac14
\kappa(s)^2\right]|\psi(s)|^2\right)
\]is the form generated by $H^0_{n,c}$, with $\dom  b^0_{n,c}=\hil_0^1(\R)$ identified with~$h_n$. 

Now we have the task of checking that if $\psi_\varepsilon\wseta\psi$ in $\mathcal E_n\subset
\LL^2(\R\times S)$, then 
\[
\liminf_{\varepsilon\to0} b^\varepsilon_{n,c}(\psi_\varepsilon)\ge  b^0_{n,c}(\psi),
\]in other words, the second main step mentioned above. By Theorem~\ref{mainTheorGamma}~{\rm i)}  we
will then conclude the strong resolvent con\-ver\-gence $H^\varepsilon_{n,c}\to H^0_{n,c}$ in~$\LL_n$,
which will complete the desired proofs. So assume the weak con\-ver\-gence
$\psi_\varepsilon\wseta\psi$.  If $\psi_\varepsilon$ does not belong to $\hil_0^1(\R\times
S)\cap\mathcal E_n$, then
$b^\varepsilon_{n,c}(\psi_\varepsilon)=\infty$, for all $
\varepsilon>0$; so we can assume that $(\psi_\varepsilon)\subset \hil_0^1(\R\times S)\cap\mathcal E_n$
and, up to subsequences, that
\[
\liminf_{\varepsilon\to0} b^\varepsilon_{n,c}(\psi_\varepsilon)=\lim_{\varepsilon\to0}
b^\varepsilon_{n,c}(\psi_\varepsilon).
\] 

Since $(\psi_\varepsilon)$ is a weakly convergent sequence, it is bounded in $\LL^2(\R\times S)$ and
we can also suppose that $\sup_\varepsilon b^\varepsilon_{n,c}(\psi_\varepsilon)<\infty$; hence, as on
page~804 of~\cite{BMT}, it follows that $(\psi_{\varepsilon})$ is a bounded sequence in
$\hil_0^1(\R\times S)$. Since  Hilbert spaces are reflexive, $(\psi_\varepsilon)$ has a subsequence,
again denoted by
$(\psi_{\varepsilon})$, so that $\psi_\varepsilon\wseta \phi$ in $\hil_0^1(\R\times S)$. Since also
$\psi_\varepsilon\wseta\psi$ in $\LL^2(\R\times S)$ it follows that $\psi=\phi$. Therefore
\[
\dot\psi_\varepsilon + \nabla_\perp \psi_\varepsilon\cdot Ry(\tau-\dot\alpha)\wseta \dot\psi +
\nabla_\perp \psi\cdot Ry(\tau-\dot\alpha)
\]and the weak lower semicontinuity of the $\LL^2$-norm (again the importance of introducing the
parameter~$c$ above), the above lower bound  $b^\varepsilon_{n,c}(\psi)\ge 1/2 \|\kappa\|_\infty^2
\|\psi\|^2$, together with the uniform con\-ver\-gences
\[
\beta_\varepsilon\to1,\quad \gamma_{\varepsilon,n}\to-\frac14 \kappa^2,
\]imply that 
\begin{eqnarray*}
\lim_{\varepsilon\to0}b^ \varepsilon_{n,c}(\psi_\varepsilon) &\ge& \mathcal G(\psi)\\
&:=&\int_{\R\times S}dsdy\left(
\left |\dot \psi+\nabla_\perp\psi\cdot R y(\tau-\dot\alpha)\right|^2 +\left(c- \frac14
\kappa^2\right)|\psi|^2\right). 
\end{eqnarray*}

If $\psi\in\dom b^0_{n,c}$, that is, $\psi=w(s)u_n(y)$, $w\in\hil^1(\R)$, then a direct substitution
infers that
$\mathcal G(wu_n)=b^0_{n,c}(wu_n)$ and so
\[
\lim_{\varepsilon\to0} b^\varepsilon_{n,c}(\psi_\varepsilon)\ge b^0_{n,c}(\psi)
\] in this case. Now we will show that, for  $\psi$ with a nonzero component in the complement of 
$\dom b^0_{n,c}$, necessarily $\lim_{\varepsilon\to0} b^\varepsilon_{n,c}(\psi_\varepsilon)=\infty$.  
In fact, if
$\psi$ does not belong to  $\dom b^0_{n,c}$ then $\|P_{n+1}\psi\|>0$ where $P_{n+1}$ is the orthogonal
projection onto 
$\mathcal E_{n+1}$. Because $\psi_\varepsilon\wseta\psi$  in $\mathcal E_n\cap\hil_0^1(\R\times S)$ it
follows that
$P_{n+1}\psi_\varepsilon\wseta P_{n+1}\psi$, and since the $\LL^2$-norm is weakly l.sc.\ we find
\[
\liminf_{\varepsilon\to0}\|P_{n+1}\psi_\varepsilon\|\ge\|P_{n+1}\psi\|>0.
\] Hence for $\varepsilon$ small enough the function $\psi_\varepsilon$ has a nonzero component
$P_{n+1}\psi_\varepsilon$ in $\mathcal E_{n+1}$ and the $\LL^2(\R\times S)$-norm of such components
are uniformly bounded from zero by $\|P_{n+1}\psi\|$.   Now,
$\sup_\varepsilon\|\psi_\varepsilon\|_\varepsilon<\infty$ and recalling that \[
\beta_\varepsilon(s,y) = 1- \xi\cdot y,\quad \xi=\varepsilon\kappa(s) z_\alpha,
\] one has
\begin{eqnarray*} b^\varepsilon_{n,c}(\psi_\varepsilon ) &=& \int_{\R\times S} dsdy 
\Big[\frac1{\beta_\varepsilon }\left|\nabla \psi_\varepsilon  \cdot(1,Ry(\tau-\dot\alpha))
\right|^2 \\ &+& \frac{\beta_\varepsilon }{\varepsilon^2}\left(|\nabla_\perp \psi_\varepsilon
|^2-\lambda_n|\psi_\varepsilon |^2\right)\Big] + c\|\psi_\varepsilon\|_\varepsilon^2 \\ &\ge& 
\int_{\R\times S} dsdy\;  \frac{\beta_\varepsilon }{\varepsilon^2}\left(|\nabla_\perp
\psi_\varepsilon |^2-\lambda_n|\psi_\varepsilon |^2\right) \\ &=& \frac1{\varepsilon^2}
\int_{\R\times S} dsdy \;\left( |\nabla_\perp \psi_\varepsilon |^2-\lambda_n|\psi_\varepsilon |^2
\right) \\ &-& 
\frac1{\varepsilon^2} \int_{\R\times S} dsdy\; (\xi\cdot y)\left( |\nabla_\perp \psi_\varepsilon
|^2-\lambda_n|\psi_\varepsilon |^2 \right).
\end{eqnarray*} Let us estimate the remanning two integrals above; for 
\[
\phi\in \hil^1_0(S)\cap [u_0,\cdots,u_{n-1}]^\perp
\] denote by $\phi^{(n)}$ the component of $\phi$ in $[u_n]$ and by $Q_{n+1}$ the orthogonal
projection onto
$[u_0,\cdots,u_n]^\perp$ in~$\hil^1_0(S)$. The first integral is positive and divergent as
$1/\varepsilon^2$ to
$+\infty$  since for $\varepsilon$ small enough
\begin{eqnarray*}
 \int_{\R\times S} &dsdy& \frac1{\varepsilon^2}\left( |\nabla_\perp \psi_\varepsilon
|^2-\lambda_n|\psi_\varepsilon |^2 \right)  \\ &=& \frac1{\varepsilon^2} \int_{\R} ds
\left(\|\nabla_\perp\psi_\varepsilon(s)\|^2_{\LL^2(S)}-\lambda_n\|\psi_\varepsilon(s)\|^2_{\LL^2(S)}
\right) 
\\ &=& \frac1{\varepsilon^2} \int_{\R} ds
\left(\|\psi_\varepsilon(s)\|^2_{\hil^1_0(S)}-(\lambda_n+1)\|\psi_\varepsilon(s)\|^2_{\LL^2(S)}
\right) 
\\ &=& \frac1{\varepsilon^2} \int_{\R} ds
\Big(\|Q_{n+1}\psi_\varepsilon(s)\|^2_{\hil^1_0(S)}+\|\psi^{(n)}_\varepsilon(s)\|^2_{\hil^1_0(S)}
\\ &-& (\lambda_n+1)\|\psi_\varepsilon(s)\|^2_{\LL^2(S)} \Big)
\\ &=&   \frac1{\varepsilon^2} \int_{\R} ds \Big(\|\nabla_\perp
Q_{n+1}\psi_\varepsilon(s)\|^2_{\LL^2(S)}+\|Q_{n+1}\psi_\varepsilon(s)\|^2_{\LL^2(S)} \\
&+&\|\nabla_\perp\psi^{(n)}_\varepsilon(s)\|^2_{\LL^2(S)} +\|\psi^{(n)}_\varepsilon(s)\|^2_{\LL^2(S)}
\\ &-& (\lambda_n+1)\|\psi_\varepsilon(s)\|^2_{\LL^2(S)} \Big)
\\ &\ge& \frac1{\varepsilon^2} \int_{\R} ds \Big(\lambda_{n+1}\|
Q_{n+1}\psi_\varepsilon(s)\|^2_{\LL^2(S)} +\lambda_n\|\psi^{(n)}_\varepsilon(s)\|^2_{\LL^2(S)}
\\ &-&\lambda_n\|\psi_\varepsilon(s)\|^2_{\LL^2(S)} \Big)
\\  &=& \frac1{\varepsilon^2} \int_{\R\times S} ds \;\left(
\lambda_{n+1}-\lambda_n\right)\|Q_{n+1}\psi_\varepsilon \|^2_{\LL^2(S)}   \\ &=& \frac{\left(
\lambda_{n+1}-\lambda_n\right)}{\varepsilon^2} \|P_{n+1}\psi_\varepsilon\|^2\ge \frac{\left(
\lambda_{n+1}-\lambda_n\right)}{\varepsilon^2} \|P_{n+1}\psi\|,
\end{eqnarray*} and, by hypothesis, $\lambda_{n+1}>\lambda_n$. The absolute value of the second
integral diverges  at most as $1/\varepsilon$; indeed, since $(\psi_\varepsilon)$ is a bounded
sequence in
$\hil_0^1(\R\times S)$ and we have
\begin{eqnarray*}
 & & \frac1{\varepsilon^2}\left|\int_{\R\times S} dsdy\; (\xi\cdot y)\left( |\nabla_\perp
\psi_\varepsilon |^2-\lambda_n|\psi_\varepsilon |^2 \right)\right|  \\ &\le& 
\frac{\|\kappa\|_\infty}{\varepsilon} \int_{\R\times S} dsdy\; |z_\alpha\cdot y|\left( |\nabla_\perp
\psi_\varepsilon |^2+\lambda_n|\psi_\varepsilon |^2 \right) 
  \\ &\le&  \frac{\|\kappa\|_\infty}{\varepsilon} \sup|z_\alpha\cdot y| \times
\max\{\lambda_n,1/\lambda_n\} \|\psi_\varepsilon\|^2_{\hil_0^1}.
\end{eqnarray*} Therefore, if $\psi$ does not belong to $\dom b_{n,c}^0$, then
\[
\liminf_{\varepsilon\to0}b^ \varepsilon_{n,c}(\psi_\varepsilon)=+\infty.
\]We have then verified the statement i) of Theorem~\ref{mainTheorGamma}, and so  Theorem~\ref{teorHn}
follows by Theorem~\ref{mainTheorGamma}v); Theorem~\ref{teorHzero} is just a particular case
with~$n=0$.

\subsection{Spectral Possibilities} There is a competition between the curvature and the twisting
terms in the effective potential \[ V_n^{\mathrm{eff}}(s) = (\tau(s)-\dot\alpha(s))^2\,C_n(S) - \frac14
\kappa(s)^2,\quad n\ge0,
\] since the curvature gives an attractive term and the twisting a repulsive one. 

This effective potential is the net result of a memory of higher dimensions that takes into account
the geometry of the  confining region. In planar  (2D) cases only the curvature term is present and
(if its not zero) a bound state does always exist.  In the spatial (3D) case, by  tuning up  the
tubes, and so the functions that define effective potentials, one finds a huge amount of spectral
possibilities for the effective Schr\"odinger operator in the reference curve
$\gamma(s)$.  Below some of them are selected; it will be assumed that $C_n(S)\ne0$ and that
$\tau,\dot\alpha$ are not necessarily bounded; this is only related to the domain of the forms and
involved operators, differently from the essential technical condition of bounded curvature. It is
worth mentioning that for lower bounded potentials $V$ (in particular for
$V_n^{\mathrm{eff}}$ above) that belong to $\LL^2_{\mathrm{loc}}(\R^k)$ the operators
$-\Delta+V$ are essentially self-adjoint when defined on~$C_0^\infty(\R^k)$; see, for instance,
Section~6.3 in~\cite{ISTQD}. 
\begin{enumerate}
\item{\bf No twisting.} In this case $ \mathcal A_n(s)=0$ and the 3D situation is quite similar to the
2D one; the effective potential $V_n^{\mathrm{eff}}(s) = - 1/4\, \kappa(s)^2$ is purely attractive and
the spectrum of $H^0_n$ has at least one negative eigenvalue. See Subsection~11.4.4 in~\cite{ISTQD}.
\item{\bf Periodic.} By choosing bounded $\kappa,\tau,\dot\alpha$ so that
$V_n^{\mathrm{eff}}(s)$ becomes periodic the resulting effective operators $H^0_n$ have purely
absolutely continuous spectra and with a band-gap structure~\cite{RS4}. Such periodicity may come from
different combinations; for instance, the tube curvature and torsion could be periodic (with the same
period) and
$\alpha(s)$ a constant function, or the tube could be straight so that $\kappa(s)=0=\tau(s)$ but the
cross section
$S$ rotates at a periodic speed
$\dot\alpha(s)$. See also \cite{ExKov}. 
\item{\bf Purely discrete.} The operators $H_n^0$ will have this kind of spectrum if 
$\lim_{|s|\to\infty}V_n^{\mathrm{eff}}(s)=\infty$ (see Section~11.5 in~\cite{ISTQD}); this happens iff
the torsion
$\tau$ or $\dot\alpha$, as well as their difference, diverge at both
$\pm\infty$. In particular $H_n^0$ will have discrete spectrum in case this limit operator is obtained
from a straight tube with growing rotation speed of the cross section such that
$\lim_{|s|\to\infty}\dot\alpha(s)=\infty$.
\item{\bf Quasiperiodic.} For one of the simplest situations select $\dot\alpha$ and $\kappa$ periodic
functions with (minimum) periods $t_\alpha>0$ and $t_\kappa>0$, respectively; if
$t_\alpha/t_\kappa$ is an irrational number we are in the case of quasiperiodic potentials.  In this
case there are many spectral possibilities that usually are very sensitive to details of the
potential. Of course one may also take
$V_n^{\mathrm{eff}}$ in the more general class of almost periodic functions; see, for instance,
\cite{CFKS}.
\item{\bf Singular continuous.} An appealing possibility is the choice of decaying  potentials
$V_n^{\mathrm{eff}}$  in the class studied by Pearson \cite{Pearson}, which leads to singular
continuous spectrum for~$H^0_n$. See also explicit examples in~\cite{KLS}.
\end{enumerate}

The reader can play with his/her imagination in order to consider tubes that give rise to previously
selected spectral types.

\subsection{Bounded Tubes}\label{subsecBTubes} Now we say something about the particular case of
bounded tubes; the goal is to recover the spectral results of~\cite{BMT}. Since the cross section~$S$
is a bounded set, the boundedness of the tube 
\[
\Omega_{\alpha,L}^\varepsilon := \left\{(x,y,z)\in\R^3: (x,y,z)=f^\varepsilon_\alpha(s,y_1,y_2), s\in[0,L],
(y_1,y_2)\in S \right \},
\] is a consequence of a bounded generating curve $\gamma(s)$ defined, say, on a compact set
$s\in[0,L]$ instead of on the whole line~$\R$ as before. In this case the negative Laplacian operator
$-\Delta_{\Omega^\varepsilon_{\alpha,L}}$ has compact resolvent and its spectrum is composed only of
eigenvalues
$\lambda^\varepsilon_{j}$, $j\in\N$; denote by 
$\psi^\varepsilon_j$ the normalized eigenfunction associated with $\lambda_j^\varepsilon$. Let
$b^\varepsilon_{L}$ and $H^\varepsilon_0(L)$ be the the corresponding sesquilinear form and
self-adjoint operator, after the ``regularizations'' and acting  in subspaces of
$\LL^2([0,L]\times S)$, suitably adapted from Subsection~\ref{subsecTH} (the same quantities as
in~\cite{BMT}).   
\begin{Theorem}\label{teorBoundTubeSpec} For each $j\in\N$ one has
\[
\lim_{\varepsilon\to0}\left(\lambda_j^\varepsilon - \frac{\lambda_0}{\varepsilon^2} -
\mu_j\right) =0,
\]where $\mu_j$ are the eigenvalues of the of the Schr\"odinger operator
\begin{eqnarray*}
\dom H^0_0(L) &=& \hil^2(0,L)\cap \hil^1_0(0,L),
\\ (H^0_0(L)\psi)(s) &=& -\ddot\psi(s) + \left((\tau(s)-\dot\alpha(s))^2C_0(S)-\frac14
\alpha^2(s)\right)\psi(s).\end{eqnarray*}  Furthermore, there are subsequences of
$f^\varepsilon_\alpha(\psi^\varepsilon_j)$ that converge to $w_j(s)u_0(y)$ in
$\LL^2([0,L]\times S)$ as $\varepsilon\to0$, where $w_j$ are the normalized eigenfunctions
corresponding to~$\mu_j$.
\end{Theorem}
\begin{proof} The proof will be an application of Proposition~\ref{GammaNorm}, with
$T=H^0_0(L)$, $T_\varepsilon=H^\varepsilon_0(L)$, $\hil_0=\{w(s)u_0(y):w\in\hil_0^1(0,L)\}$. Let 
$b^0_L$ be the form generated by $H_0^0(L)$. Previously discussed results in the case of unbounded
tubes apply also here and they show that 
\[ b^\varepsilon_L\sgconv b^0_L,\qquad b^\varepsilon_L\wgconv b^0_L,
\] that is, item~a) of Proposition~\ref{GammaNorm} holds in this setting. Since $H^0_0(L)$ has compact
resolvent, item~b) in that proposition follows at once.

 Finally, for each $\varepsilon>0$ the two hypotheses, $(\psi_\varepsilon)$ is bounded in
$\LL^2(\R\times S)$ and $b^\varepsilon_L(\psi_\varepsilon)$ is bounded, imply that (see page~804
of~\cite{BMT}) 
$(\psi_\varepsilon)$ is a bounded sequence in $\mathcal K=\hil_0^1([0,L]\times S)$. By
Rellich-Kondrachov Theorem the space $\mathcal K$ is compactly embedded in $\LL^2([0,L]\times S)$ (due
to the boundedness of
$[0,L]\times S$), and so item~c) of Proposition~\ref{GammaNorm} holds. By that proposition
$H^\varepsilon_0(L)$ converges in  the norm resolvent sense to
$H^0_0(L)$ in $\hil_0$, and it is well known that the spectral assertions in
Theorem~\ref{teorBoundTubeSpec} follow by this kind of con\-ver\-gence, that is, the con\-ver\-gence
of eigenvalues of $H^\varepsilon_0(L)$ to eigenvalues of
$H^0_0(L)$ as well as the assertion about con\-ver\-gence of eigenfunctions.  Taking into account that
in the construction of
$H^\varepsilon_0(L)$ there was  the ``regularization'' subtraction of
$(\lambda_0/\varepsilon^2)\|\psi\|^2_\varepsilon$ from the original form of the Laplacian
$-\Delta_{\Omega^\varepsilon_{\alpha,L}}$, the conclusions of Theorem~\ref{teorBoundTubeSpec} follow.
\end{proof}

\begin{Remark} Theorem~\ref{teorBoundTubeSpec} makes clear the mechanism behind the spectral
approximations in case of bounded tubes, that is, the powerful {\em norm resolvent con\-ver\-gence} is
in action!
\end{Remark}

\begin{Remark}
 Although we expect that for Theorem~\ref{teorHn} the norm resolvent con\-ver\-gence takes place, we
were not able to prove it; at the moment, to get norm convergence we need a combination of $\Gamma$-con\-ver\-gence and
compactness of the tube (as in Theorem~\ref{teorBoundTubeSpec}). A very simple example indicates how
subtle those properties can be combined and that our expectations might  be wrong.  
 
 Consider the sequence of multiplication operators $T_n\psi(x)=x\psi(x)/n$ and $T=0$. In the space
${\LL}^2(\R)$, dominated con\-ver\-gence implies that $T_n$ converges to $T$ in the strong resolvent
sense. Now,
$\sigma(T_n)=\R$, for all~$n$, while $\sigma(T)=\{0\}$; thus, $T_n$ does not converge in the norm
resolvent sense to
$T$. However, for the same operator actions  in $\LL^2[0,1]$ one gets that $T_n$ converges to $T$ in
the norm resolvent sense (due to the compactness of $[0,1]$).
\end{Remark}

\begin{Remark} It is also possible to consider semi-infinite tubes, that is, $s\in[0,\infty]$. In this
case all previous constructions apply and the limit operator has the expected action but with
Dirichlet boundary condition at~zero. The details are similar to the arguments previously discussed
here and in~\cite{BMT} and will be omitted.
\end{Remark}

\section{Broken-Line Limit}\label{sectBrokenLine} In this section we discuss the operators
$H^0_n$, defined on a spacial curve $\gamma(s)$ with compactly supported curvature, that approximate
another singular limit, now given by two infinite straight edges with one vertex at the origin. The
angle between the straight edges is $\theta$ and is kept fixed during the approximation process. In
case of planar  curves this problem has been considered in
\cite{DellAntTen,ACF} and a variation of it in~\cite{CacciaExner}, and those authors had at hand
explicitly expressions for the resolvents of the Hamiltonians as integral kernels.   

This geometrical broken line is a simple instance of a {\em quantum graph} and the main question is
about the boundary conditions  that is selected at the vertex in the con\-ver\-gence process; that
is expected to be the physical boundary conditions. We refer to the above cited references for more physical
and mathematical details.   Note that in this work we have restricted ourselves to first confine the
quantum system from the tube to the curve, and then take the broken-line limit. In the planar curve cases
both limits are taken together (as in 
\cite{DellAntTen,ACF,CacciaExner}) and with no reference to $\Gamma$-con\-ver\-gence; but since it may
involve sequences of operators that are not uniformly bounded from below
\cite{DellAntTen}, it is not clear that in our 3D setting  we could address both limits together by
using
$\Gamma$-con\-ver\-gence; this seems to be an interesting open problem.  

Of course a novelty here is the possibility of quantum twisting in the effective potentials
\[ V_n^{\mathrm{eff}}(s) = (\tau(s)-\dot\alpha(s))^2\,C_n - \frac14 \kappa(s)^2,\quad n\ge0,
\]since in the plane cases \cite{DellAntTen,ACF} only the curvature term $-\kappa(s)^2/4$ is present.
Another interesting point is the dependence of the effective potential on the
$(n+1)$th sector spanned by the eigenvector $u_n$ of the Laplacian restricted to the cross
section~$S$. Thus the limit operator depends on~$n$ and since this additional term $\mathcal A_n(s)$
is positive we have a wide range of possibilities in the 3D case, that is, not just an attractive
potential as in~2D.    From the technical point of view we will follow closely the proof of Lemma~1
in~\cite{ACF}, which uses results of~\cite{BGW}. However, differently than the planar situation, the
condition
\[
\la V_n^{\mathrm{eff}}\ra:= \int_\R ds\,V_n^{\mathrm{eff}}(s)\ne0
\]may not hold in  3D, but we will see that the same proof can be adapted to the case  $\la
V_n^{\mathrm{eff}}\ra=0$ by using results of~\cite{BGK}; the boundary conditions at the vertex  depend
explicitly on the curvature and twisting.    

We will be rather economical in the proofs below, since
we do not intend to just repeat whole parts of published works; we are sure that from the statements
below and references to papers and specific equations, the interested reader will have no special difficulties
in filling out the missing details.   

Assume that the curvature, torsion and the speed of the rotation
angle $\dot\alpha$ are compactly supported in $(-1,1)$ and scale them as
\[
\kappa_\delta(s):=\frac1\delta \kappa\left( \frac s\delta \right), \quad
\tau_\delta(s):=\frac1\delta \tau\left( \frac s\delta \right), \quad
\dot\alpha_\delta(s):=\frac1\delta \dot\alpha\left( \frac s\delta \right), 
\]and the continuations of the half-lines to the left and to the right of that support joint at the
origin with an angle~$\theta$; this angle is exactly the integral 
\[
\theta=\int_\R ds\, \kappa_\delta(s) =\int_\R ds\, \kappa(s),\quad \forall \delta>0.
\] Of course, as above, the curve $\gamma$ is supposed to be smooth and without self-intersection.
Here we consider only the above scales.  Our concern now is to study the limit $\delta\to0$ of the
families of operators
\[ (H_n^0(\delta)\psi)(s) = -\ddot\psi(s) + V_{n,\delta}^{\mathrm{eff}}(s)\psi(s), \quad \dom
H^0_n(\delta) =
\hil^2(\R),\quad n\ge0,
\]where
\[ V_{n,\delta}^{\mathrm{eff}}(s):= (\tau_\delta(s)-\dot\alpha_\delta(s))^2\,C_n - \frac14
\kappa_\delta(s)^2.
\]

It turns out that this limit $\delta\to0$ is related to the low energy expansion of the resolvent
$R_{k^2}(H_n^0)$, $\Imm k>0$, as explained on page~8 of~\cite{ACF}. The operator
$H^0_n$ is said to have a {\em resonance at zero} if there exists $\psi_r\in \LL^\infty(\R)$,
$\psi\notin\LL^2(\R)$, such that $H_n^0\psi_r=0$ in the sense of distributions; in this case
$\psi_r$ can be chosen real and is unique (as a subspace). 
 Since all $V_n^{\rm eff}$ have compact support, one has $\int_\R ds\, e^{as}|V_n^{\rm
eff}(s)|<\infty$ for some
$a>0$, which is a technical condition necessary for what follows
\cite{ACF,BGK,BGW}. Now we consider two complementary cases: $\la V_n^{\rm eff}\ra \ne0$ and
$\la V_n^{\rm eff}\ra =0$.

\subsection{$\la V_n^{\rm eff}\ra \ne0$} Assume that this condition holds. In this case we may
directly  apply Lemma~1 in~\cite{ACF}, which employs results of~\cite{BGW}, to obtain: 
\begin{Proposition}\label{propVneZero} {\rm (a)} If  $H^0_n$ has no resonance at zero, then 
$H^0_{n,\delta}$ converges in the norm resolvent sense, as $\delta\to0$, to the one-dimensional
Laplacian
$-\Delta^D$ with Dirichlet boundary condition at the origin, that is, 
\begin{eqnarray*}
\dom (-\Delta^D) &=& \left\{ \psi\in \hil^1(\R)\cap\hil^2(\R\setminus\{0\}): \psi(0)=0 \right\},
\\ (-\Delta^D\psi)(s) &=& -\ddot\psi(s).
\end{eqnarray*}

{\rm (b)} If  $H^0_n$ has a resonance at zero, then  $H^0_{n,\delta}$ converges in the norm resolvent
sense, as
$\delta\to0$, to the one-dimensional Laplacian $-\Delta^r$ given by
\begin{eqnarray*}
\dom (-\Delta^r) &=& \big\{ \psi\in \hil^2(\R\setminus\{0\}):
(c_1^n+c_2^n)\psi(0^+)=(c_1^n-c_2^n)\psi(0^-),
\\ &  &(c_1^n-c_2^n)\dot\psi(0^+)=(c_1^n+c_2^n)\dot\psi(0^-) \big\},
\\ (-\Delta^r\psi)(s) &=& -\ddot\psi(s),
\end{eqnarray*}where 
\begin{eqnarray*} c_1^n &=& \frac1{2\la V_n^{\rm eff}\ra} \int_{\R\times \R}dsdy\; V_n^{\rm
eff}(s)\,|s-y|\,V_n^{\rm eff}(y)\psi_r(y),
\\ c_2^n &=&  -\frac1{2} \int_\R ds\; sV_n^{\rm eff}(s)\psi_r(s).
\end{eqnarray*}Moreover, $c_1^n$ and $c_2^n$ do not vanish simultaneously.
\end{Proposition}

\subsection{$\la V_n^{\rm eff}\ra =0$} Assume that this condition holds. Now we cannot apply directly
Lemma~1  in~\cite{ACF}, but by invoking results of~\cite{BGK} we can check that the proof of such
Lemma~1 may be replicated to conclude: 
\begin{Proposition}\label{propVequalZ} {\rm (a)} If  $H^0_n$ has no resonance at zero, then 
$H^0_{n,\delta}$ converges in the norm resolvent sense, as $\delta\to0$, to the one-dimensional
Laplacian
$-\Delta^D$ with Dirichlet boundary condition at the origin.  

{\rm (b)} If  $H^0_n$ has a resonance at zero, then  $H^0_{n,\delta}$ converges in the norm resolvent
sense, as
$\delta\to0$, to the one-dimensional Laplacian $-\Delta^r$, as in Proposition~\ref{propVneZero}{\rm
(b)}, but now
\begin{eqnarray*} c_1^n &=& \frac1{2W} \int_{\R^3}dsdxdy\; V_n^{\rm eff}(s)\,|s-x|\,V_n^{\rm
eff}(x)\,|x-y|\,V_n^{\rm eff}(y)\psi_r(y),
\\ c_2^n &=&  -\frac1{2} \int_\R ds\; sV_n^{\rm eff}(s)\psi_r(s), \\ W &=& \int_{\R^2}dsdy\; V_n^{\rm
eff}(s)\,|s-y|\,V_n^{\rm eff}(y)\,>0.
\end{eqnarray*}Moreover, $c_1^n$ and $c_2^n$ do not vanish simultaneously.
\end{Proposition}

Note the different expressions for the parameter $c_1^n$ from the case $\la V_n^{\rm eff}\ra
\ne0$; in both Propositions~\ref{propVneZero} and~\ref{propVequalZ}, the expressions for
$c_1^n,c_2^n$ were obtained by working with relations in references~\cite{BGW} and~\cite{BGK},
respectively. In order to replicate the proof of the above mentioned Lemma~1, it is enough to check 
some key properties that can be found spread along reference~\cite{BGK}; there is a complete
parallelism between both cases, although the expressions defining the involved quantities are
different (that was a chief contribution of~\cite{BGK}).   In what follows we indicate what are such
properties, where their versions in case $\la V_n^{\rm eff}\ra =0$ can be found in~\cite{BGK} and we
use the notation of~\cite{ACF,BGK} without explaining the meaning of some of the symbols employed
(e.g.,
$t_j, M_j,\phi_0,\cdots$). Unfortunately a short explanation of the involved symbols will not be very
helpful to the understanding of the large amount of involved technicalities; at any rate, they are not necessary to state the above results, they can be
easily found in the references and the equations in~\cite{ACF,BGK} we shall use in the proof below
will be explicitly indicated.  
\begin{proof} Introduce the functions 
\[ v=|V_n^{\rm eff}|^{1/2},\quad u=|V_n^{\rm eff}|^{1/2}(\mathrm{sgn}V_n^{\rm eff}),
\] so that $V_n^{\rm eff}=vu$ and $\la V_n^{\rm eff}\ra = (v,u)$ (inner product in
$\LL^2(\R)$). The properties needed for the proof of Proposition~\ref{propVequalZ}(a) appear in
equation~(25) of~\cite{ACF}, that is,
\[ (v,t_0u)=0,\quad ((\cdot)v,t_0u)=(v,t_0u(\cdot))=0,\quad (v,t_1u)=-2.
\]  The first and the third ones can be found in equation~(3.83) of~\cite{BGK}, while the second one
is obtained by combining  equations~(2.8) and~(3.98) of that work.  

For the proof of Proposition~\ref{propVequalZ}(b) one need to check equations~(17) and~(18)
of~\cite{ACF}; equation~(18) reads
\begin{eqnarray*} ((\cdot)v,t_{-1}u(\cdot))&=&\frac{2(c_2^n)^2}{(c_1^n)^2+(c_2^n)^2},\\
((\cdot)v,t_0u)&=&\frac{2c_1^nc_2^n}{(c_1^n)^2+(c_2^n)^2},\\
(v,t_1u)&=&-\frac{2(c_2^n)^2}{(c_1^n)^2+(c_2^n)^2};
\end{eqnarray*}these relations are found in equations~(4.16), (4.15) and~(3.91) of~\cite{BGK},
respectively.  Now equation~(17) of~\cite{ACF} reads
\[ t_{-1}u=0,\quad t_{-1}^*v=0,\quad (v,t_0u)=0.
\] The third relation follows from equation~(3.90) of~\cite{BGK}, and their equation~(3.93) implies
(recall we are using their notation)
\[ t^*_{-1}v = (\mathrm{sgn}V_n^{\rm eff})t_{-1}(\mathrm{sgn}V_n^{\rm eff})v=(\mathrm{sgn}V_n^{\rm
eff})t_{-1}u=0,
\]that is, we have got the second relation by accepting that the first one holds. Now we show how to
derive the first one from~\cite{BGK}. By equations~(3.45) and~(3.5) in \cite{BGK} it is found that
\begin{eqnarray*} t_{-1} &=& -c_0 P_0\hat{Q}=-c_0P_0(\Id-\hat{P}) \\ &=&-c_0P_0\left(
\Id-\frac1c M_0P\right) \\ &=&  -c_0P_0 + \frac{c_0}{c} M_0P,
\end{eqnarray*}with $P(\cdot) = (v,\cdot)u$ and (see also equations~(3.2) and~(3.3)  in \cite{BGK})
\[ P_0(\cdot)=(\hat{\phi_0},\cdot)\phi_0, \quad \hat{\phi_0}= (\mathrm{sgn}V_n^{\rm eff})M_0\phi_0.
\] 

The proof finishes as soon as we check that $P_0u=0$ and $Pu=0$.   By the hypothesis on the potential
we have
\[ Pu = (v,u)u = \la V_n^{\rm eff}\ra u=0,
\]and 
\begin{eqnarray*} P_0u &=& (\hat{\phi_0},u)\phi_0 = ((\mathrm{sgn}V_n^{\rm eff})M_0\phi_0,u)\phi_0 
\\ &=& (M_0\phi_0,(\mathrm{sgn}V_n^{\rm eff})u)\phi_0=(M_0\phi_0,v)\phi_0
\end{eqnarray*}which vanishes by equation~(3.10) in~\cite{BGK}.
\end{proof}

\end{document}